\documentclass[letterpaper,11pt]{article}

\usepackage[utf8]{inputenc}
\usepackage{pstricks,pstricks-add} 
\usepackage{amsthm}
\usepackage{amsmath,amsfonts,amssymb,latexsym,cite}

\usepackage{enumerate}
\usepackage{hyperref} 
\usepackage{graphicx} 
\graphicspath{{./figures/}} 
\usepackage[font={it}]{caption}
\usepackage[margin=0.05in]{subcaption}
\usepackage{color}
\usepackage{enumitem}
\usepackage{epsfig}
\usepackage{graphicx}

\usepackage{latexsym}
\usepackage{amsmath}
\usepackage{amssymb}
\usepackage{soul}
\usepackage[toc,page]{appendix}

\usepackage[normalem]{ulem}

\newtheorem{theorem}{Theorem}[section]
\newtheorem{lemma}[theorem]{Lemma}
\newtheorem{claim}[theorem]{Claim}

\newtheorem{corollary}[theorem]{Corollary}
\newtheorem{definition}[theorem]{Definition}

\newcommand{\del}{\delta}

\newcommand{\polylog}{\mbox{polylog}}

\newcommand{\sj}[1]{#1}
\newcommand{\zxx}[1]{#1}
\newcommand{\zt}[1]{#1}
\newcommand{\mikel}[1]{#1}

\newcommand{\zx}[1]{#1}

\newcommand{\alfq}[1][ ]{\alpha_{q}\left( \bar{p}_{#1}\right) }
\newcommand{\alfexpq}[1][ ]{1-\frac{2q}{q-1}\left(p-\bar{p}_{#1}\right)-\frac{q}{q-1}\pstar}

\newcommand{\capopq}{\underset{{\bar{p}\in\left[ 0,p\right] }}{\min}\left[ \alfq\left( 1-H_{q}\left( \frac{\bar{p}}{\alfq}\right) \right) \right]}

\newcommand{\capq}[1][0]{\alfq[#1]\left( 1-H_{q}\left( \frac{\bar{p}_{#1}}{\alfq[#1]}\right) \right) }

\newcommand{\leftcw}{t}
\newcommand{\rightcw}{h}
\newcommand{\mt}{t}

\newcommand{\pref}[1]{{#1}_{\leftcw}}
\newcommand{\suff}[1]{{#1}_{\rightcw}}

\newcommand{\phat}{\hat{p}}
\newcommand{\ptil}{\tilde{p}}
\newcommand{\phatt}{\pref{\phat}}
\newcommand{\pbar}{\bar{p}}
\newcommand{\ppp}{p}

\newcommand{\pstar}{p^{\star}}

\newcommand{\chunknum}{1/\echunk} 
\newcommand{\chunk}{n\echunk}
\newcommand{\choiceschunk}{\left\lbrace \chunk,2\chunk,\cdots,n-\chunk\right\rbrace }

\newcommand{\setchend}{\mathcal{T}}

\newcommand{\intvera}{0,\frac{q-1}{q}}
\newcommand{\intverr}{0,\frac{q-1}{2q}}

\newcommand{\msg}{\mathcal{U}}
\newcommand{\secr}{\mathcal{S}}
\newcommand{\mlist}{\mathcal{L}}
\newcommand{\exlist}{\mathcal{L}{\left(m\right) }}
\newcommand{\clist}{{\exlist}}
\newcommand{\mlists}{{L}}
\newcommand{\clists}{{L{\left(m\right)}}}

\newcommand{\mbfx}{\mathbf{X}}
\newcommand{\mbfy}{\mathbf{Y}}
\newcommand{\mbfu}{\mathbf{U}}

\newcommand{\mx}{\mathbf{x}}
\newcommand{\my}{\mathbf{y}}

\newcommand{\cwlist}{\left\lbrace \mathbf{x}_{1},\mathbf{x}_{2},\cdots,\mathbf{x}_{\clists}\right\rbrace }

\newcommand{\echunk}{\theta}
\newcommand{\echval}{\frac{\erate^{2}}{9q^{2}}}
\newcommand{\egood}{q^{-n\echunk^4}}
\newcommand{\erate}{\epsilon}
\newcommand{\elist}{\frac{\erate}{4}}

\newcommand{\errdist}{\frac{(n-t)\erate^2}{9q^2}}

\newcommand{\epht}[1][\mt]{\frac{\erate^{2}}{9q^{2}\alpha_q^{2}\left(\bar{p}_{#1}\right)}}

\newcommand{\lat}[1][t]{\lambda_{#1}}

\newcommand{\mcode}{\mathcal{C}}

\newcommand{\megacode}[1][k]{\mathcal{C}_{1}\circ\mathcal{C}_{2} \circ\cdots\circ \mathcal{C}_{#1}}
\newcommand{\megacodw}[1][k]{\mathcal{C}_{1}\left( m,s_1\right) \circ\mathcal{C}_{2}\left( m,s_2\right) \circ\cdots\circ \mathcal{C}_{#1}\left( m,s_{#1}\right) }

\newcommand{\rmegacode}[1][k]{\mathcal{C}_{#1+1}\circ\mathcal{C}_{#1+2}\circ\cdots\circ \mathcal{C}_{\chunknum}}
\newcommand{\rmegacodw}[1][k]{\mathcal{C}_{#1+1}\left( m,s_{#1+1}\right) \circ\mathcal{C}_{#1+2}\left( m,s_{#1+2}\right) \circ\cdots\circ \mathcal{C}_{\chunknum}\left( m,s_{\chunknum}\right) }

\newcommand{\mpm}{m^{\prime}}

\newcommand{\tstar}{t^{\star}}

\newcommand{\kstar}{k^{\star}}

\newcommand{\blen}{b}

\newcommand{\epsbar}{\bar{\epsilon}}

\newcommand{\dl}{\Delta}
\newcommand{\keyspace}{\left[q^{nS}\right]}
\newcommand{\msgspace}{\left[q^{nR}\right]}

\newcommand{\qaryset}[1][ ]{\left\lbrace0,1,\cdots,q-1\right\rbrace ^{#1}}

\newcommand{\blistord}{O\left( \frac{1}{\erate}\right) }

\newcommand{\hmd}[2]{d_{H}\left(#1,#2\right)}

\newcommand{\xbab}{\mathbf{x}_{b}} 
\newcommand{\xpus}{\mathbf{x}_{p}}
\newcommand{\ybab}{\mathbf{y}_{b}}
\newcommand{\ypus}{\mathbf{y}_{p}}

\newcommand{\ebab}{\frac{\epsilon}{2}}

\theoremstyle{remark}
\newtheorem{remark}[theorem]{Remark}

\begin{document}
	\title{The Capacity of Online (Causal) $q$-ary Error-Erasure Channels}
	\author{
		Z. Chen
		\thanks{Department of Electrical and Computer Engineering, University of Maryland, College Park, \texttt{chenztan@umd.edu}}
		\and 
		S. Jaggi
		\thanks{Department of Information Engineering, The Chinese University of Hong Kong, \texttt{jaggi@ie.cuhk.edu.hk}}
		\and 
		M. Langberg
		\thanks{Department of Electrical Engineering, State University of New York at Buffalo, \texttt{mikel@buffalo.edu}}
	}
	\date{}
	\maketitle
	

\begin{abstract}
	\noindent%
	In the $q$-ary online (or ``causal'') channel coding model, a sender wishes to communicate a message to a receiver by transmitting a codeword $\mathbf{x} =(x_1,\ldots,x_n) \in \{0,1,\ldots,q-1\}^n$ symbol by symbol via a channel limited to at most $pn$ errors and/or $\pstar n$ erasures. 
	The channel is ``online'' in the sense that at the $i$th step of communication the channel decides whether to corrupt the $i$th symbol or not based on its view so far, i.e., its decision depends only
	on the transmitted symbols $(x_1,\ldots,x_i)$. 
	This is in contrast to the classical adversarial channel in which the corruption is chosen by a channel that has a full knowledge on the sent codeword $\mathbf{x}$. 
	
	In this work we study the capacity of $q$-ary online channels for a combined corruption model, in which the channel may impose at most $pn$ {\em errors} and at most $\pstar n$ {\em erasures} on the transmitted codeword. The online channel (in both the error and erasure case) has seen a number of recent studies which present both upper and lower bounds on its capacity. In this work, we give a full characterization of the capacity as a function of $q,p$, and $\pstar$. 
\end{abstract}

\newpage
\setcounter{page}{1}

\section{Introduction}
Reliable communication over different types of channels has been extensively studied in electrical engineering and computer science. One frequently used communication channel model is the binary erasure channel, in which a bit (a zero or one) is either transmitted intact or erased. Specifically, an erased bit is a visible error, denoted by a special symbol $\Lambda$, which can be identified directly by a receiver. Another frequently studied channel model is the binary bit-flip channel, where bits can be flipped to their complement. Further generalization of channel alphabet to size of $q\geq 2$ leads to general $q$-ary channels. 

There are two broad approaches to model (erasure or error) corruptions imposed by the channel. Shannon's approach is to model the channel as a stochastic process; Hamming's approach is a combinatorial approach to model the channel by an adversarial process that can manipulate parts of the transmitted codeword arbitrarily, subject only to a limit on the number of corrupted symbols. 

It is interesting to further classify the Hamming model for an adversarial channel in terms of the adversary's knowledge of the codeword.
Some examples include the standard adversarial channel (also referred to here as the {\em omniscient} adversary), e.g., \cite{gilbert1952comparison,varshamov1957estimate,mceliece1977new}, the {\em causal} (or {\em online}) adversary, e.g., \cite{dey2013codes,langberg2009binary,haviv2011beating,dey2012improved,bassily2014causal,zitan2015causal}, and the {\em oblivious} adversary, e.g.,\cite{lapidoth1998reliable,Langberg:08,guruswami2010codes}; from the strongest adversarial power to weakest. In one extreme, the omniscient adversarial model (a.k.a. the classical adversarial model) assumes that the channel has full knowledge of the entire codeword, and based on this knowledge, the channel can maliciously decide how to corrupt the codeword. In the other extreme, the oblivious adversarial model is a model in which the channel is clueless about the codeword and generates corruptions in a manner that is independent of the codeword being transmitted. The causal adversarial model is an intermediate model between the two extremes, in which the channel decides whether to tamper with a particular symbol of the codeword based only on the symbols transmitted so far. 
{There are significant differences between the different adversarial models classified above (with respect to their capacity). We elaborate on these differences shortly.}

In this work we focus on causal adversaries, and study reliable communication over \textit{$q$-ary causal adversarial channels}. Specifically, we consider the following communication scenario. A sender (Alice) wishes to transmit a message $m\in\msg$ to a receiver (Bob) over a $q$-ary causal adversarial channel by encoding $m$ into a codeword $\mathbf{x}=(x_1,x_2,\cdots,x_n)\in\qaryset[n]$ of length $n$. However, the channel is governed by a causal adversary (Calvin), who can observe $\mathbf{x}$ and impose up to a $pn$ errors and $\pstar n$ erasures. More importantly, Calvin decides whether to tamper with the $i$-th symbol of the codeword based only on the symbols $(x_1,x_2,\cdots,x_i)$ transmitted thus far. 
Roughly, if $q^{nR}$ distinct messages can be sent using codewords of length $n$, we say that a code achieves rate $R$. We are interested in the maximum achievable rate $R$, which is the capacity $C$ of the channel. (See Section~\ref{sec:model} for precise definitions.)

\subsection{Our Results}
In this work we characterize the capacity of $q$-ary causal channels as a function of alphabet size $q$, error capability $p$, and erasure capability $\pstar$. Specifically, we propose and analyze
an attack strategy similar to those for the binary cases~\cite{dey2012improved, bassily2014causal}  (to be described in detail shortly), which gives an upper bound on the capacity, and a coding scheme similar to the one given in \cite{zitan2015causal}, which implies a lower bound on the capacity matching our upper bound. 
Our main result can be summarized by the following theorem.

\begin{theorem}
	\label{thm:main:capacity}
	The capacity $C$ of $q$-ary causal adversarial channels with symbol errors and erasures is
	\begin{equation}
		C=\left \{
		\begin{array}{lc}
			\capopq, \mbox{ }& p\in\left[0,\frac{q-1}{2q}\right],\pstar\in\left[0,\frac{q-1}{q}\right],\text{and }p+\pstar\leq\frac{q-1}{q}, \\
			0,\mbox{ } & \text{otherwise,}
		\end{array}
		\right .
	\end{equation} where $\alfq=\alfexpq$.
\end{theorem}

In fact, as direct by-products of the analysis of our coding scheme, we can show that even if Calvin has ``small'' lookahead, the capacity is essentially unchanged. More precisely, if for any constant $\epsilon > 0$, Calvin decides whether to tamper with the $i$-th symbol of the codeword based only on the symbols $(x_1,x_2,\cdots,x_{j})$, where $j = \min\{n,i+n\epsilon\}$, then the capacity of the corresponding ``$n\epsilon$-lookahead is at most $f(\epsilon)$ less than the corresponding $C$ we show in Theorems~\ref{thm:main:capacity} above (for some continuous $f$). We provide a rough argument in support of this claim in the Remark at the end of Section~\ref{sec:code-analysis}.

\subsection{Previous Work}

We start by briefly summarizing the state-of-the-art for erasure and error adversarial channels, for both omniscient and oblivious adversaries.
The optimal rate of communication over binary omniscient adversarial channels (for both erasure and error) are long standing open problems in coding theory. The best known lower bounds for the problems derive from the Gilbert-Varshamov codes (the GV bound)~\cite{gilbert1952comparison,varshamov1957estimate}, and the tightest upper bounds (the MRRW bounds) from the work by McEliece {\it et al.}~\cite{mceliece1977new}. 

The literature on Arbitrarily Varying Channels (AVCs, e.g., \cite{lapidoth1998reliable}) implies that the capacity of the binary oblivious adversarial error channel is $1-H(p)$, and that of oblivious adversarial erasure channels is $1-p$; these match the well-known capacities of the corresponding ``random noise'' channels with bits flipped or erased Bernoulli($p$), but are attainable even for noise patterns that can be chosen (up to an overall constraint of a $p$-fraction corruptions) by an adversary with full knowledge of the codebook, but no knowledge of the actually transmitted codeword.\footnote{In fact, it can even be shown that if Alice is allowed to use {\it stochastic encoding} -- choosing one of multiple possible codewords randomly for each message she wants to transmit -- then even for a {\it maximal probability of error} metric, a vanishingly small probability of error can be attained by capacity achieving codes. That is, there exists a sequence of codes whose rates asymptotically achieve the corresponding capacity, and such that for every message transmitted by Alice and for every corruption pattern imposed by Calvin, can be decoded correctly by Bob for ``most'' codewords corresponding to that message.} An alternate proof of the capacity of the binary oblivious bit-flip channel was presented in~\cite{Langberg:08} by Langberg, and a computationally efficient scheme achieving this rate was presented in~\cite{guruswami2010codes} by Guruswami and Smith.

We now turn to the causal setting.
As a causal adversary can never do better than an omniscient adversary and does at least as well as an oblivious one, the upper bounds on capacity for oblivious adversaries specified above act as upper bounds for the causal case as well;
and the lower bounds on capacity for omniscient adversaries act as lower bounds for the causal case.
For the binary causal adversarial bit-flip channel both bounds were improved.
Specifically, the first nontrivial upper bound $\min\left\lbrace 1-H(p),(1-4p)^{+}\right\rbrace $ was given by Langberg {\it et al.}~\cite{langberg2009binary}, and later, the tightest upper bound was given by the continuing work of Dey {\it et al.}~\cite{dey2012improved,dey2013upper}.
The best lower bound was described by Haviv and Langberg~\cite{haviv2011beating} which slightly improves over the GV bound. 
For the binary causal adversarial erasure channel the trivial upper bound of $1-p$ was improved to $1-2p$ 
by Bassily and Smith~\cite{bassily2014causal} who also present improved lower bounds that separate the achievable rate for causal adversarial erasures from the rates achievable for omniscient adversarial erasures.
Recently, the capacities for binary causal adversarial erasures and errors were fully characterized by \cite{zitan2015causal} which we demonstrate equals $C$ of Theorem~\ref{thm:main:capacity} for the case where $q=2$ and $\ppp=0$, and the case where $q=2$ and $\pstar=0$, respectively.


Related results include the study of binary delayed adversaries by Dey {\it et al.}~\cite{dey2010coding} who provide a characterization of the capacity in the case of ``delays'' $d$ which are an arbitrarily small (but constant) fraction of the code block length $n$.\footnote{While not presented in that work, the techniques of~\cite{dey2010coding} can be used to show that the same capacity holds even if the delay is $\polylog(n)$ rather than $d={\cal O}(n)$.} The value $d$ here corresponds to an adversarial model in which the decision of whether or not to corrupt the $i$th codeword bit depends only on $(x_1,\ldots,x_{i-d})$ (and the overall constraint on the number of bits that can be corrupted).
It is interesting to note that, in this case as well as the oblivious one, the capacity of the bit-flip and bit-erasure channels matches the corresponding random noise capacities (of $1-H(p)$ and $1-p$). On the other hand, as mentioned, the causal and $n\epsilon$ lookahead settings have strictly lower, but approximately matching, capacities. This seems to imply that the knowledge of the present is critical for Calvin to significantly depress the capacity below the random noise capacity. 

While the above discussion relates to the problem of binary alphabets, the work of Dey {\it et al.}~\cite{dey2013codes} considered ``large alphabet channels'' (in which the alphabet size is ``significantly larger'' than the block-length $n$) with causal symbol errors.\footnote{The capacity of large alphabet causal symbol erasures is essentially the same as that of omniscient large alphabet symbol erasures, which in turn equals the capacity of random symbol erasures. Such rates can be directly attained by Reed-Solomon codes, and matching converses obtained by Calvin merely randomly erasing $pn$ symbols.} A complete capacity characterization was presented (with corresponding computationally efficient codes attaining capacity), which demonstrated that the capacity of this problem equals $1-2p$, which is the same as the capacity of an omniscient adversary (attained by Reed-Solomon codes, and impossibility of higher rates by the Singleton bound). This demonstrates that the penalty imposed by the causality condition on Calvin diminishes with increasing alphabet size.

Also related to this work is the study of Mazumdar~\cite{mazumdar2014capacity} in which the capacity of memoryless channels where the adversary makes his decisions based only on the value of the currently transmitted bit is addressed. We note that the causal model is also a variant of the AVC model~\cite{blackwell1960capacities,lapidoth1998reliable}, however previous works on AVCs with capacity characterizations do not relate directly to the study at hand on causal adversaries.

\subsection{Proof Technique}

To prove Theorem~\ref{thm:main:capacity} we demonstrate two results: a converse (by analyzing an attack strategy similar to that presented in \cite{dey2012improved, bassily2014causal,dey2013upper}) and a coding scheme (that follows the lines of that presented  in \cite{zitan2015causal}). Our major novelty lies in extending the proof techniques to hold for $q$-ary causal adversarial channels for general $q$ where the adversary is able to impose both errors and erasures on codewords.
Throughout, we denote the encoder by Alice, the decoder by Bob, and the adversarial causal jammer by Calvin.

\subsubsection{Converse}

To prove Theorem~\ref{thm:main:capacity} we must present a strategy for Calvin that does
not allow communication at rate higher than $C$ (no matter which
encoding/decoding scheme is used by Alice and Bob).
Specifically, the strategy we present will allow Calvin to enforce a constant
probability of error bounded away from zero whenever Alice and Bob communicate
at rate higher than $C$.
Calvin uses a two-phase {\em babble-and-push} strategy.

In the first phase Calvin ``babbles'' by behaving
like a $q$-ary symmetric channel in which at most $\bar{p}n$ symbols are changed.
There is an adversarial attack of Calvin for any $\bar{p} \leq p$, but it is ``strongest'' for an optimal $\bar{p}$ that depends on the setting of $q$, $p$, and $\pstar$.
This fact is what accounts for the minimization in the capacity term given in Theorem~\ref{thm:main:capacity}.
The value of $\bar{p}$ also determines the {\em length}, denoted here by $b$, of the babble phase, namely when Calvin stops behaving like a $q$-ary symmetric channel and starts his second ``push" phase.
As $\bar{p}$ is taken to be at most $p$, in this first phase, Calvin only uses his error capabilities (and does not erase any symbols).

In the second phase of $n - b$ channel uses, Calvin randomly selects a codeword 
from Alice and Bob's codebook which is {\em consistent} with what Bob has 
received so far.  Namely, a codeword that from Bob's perspective may have been transmitted (when taking into account Calvin's attack).
Calvin then ``pushes'' the remaining part of Alice's codeword towards his selected codeword.
The push phase includes both errors and erasures on Calvin's behalf.
Specifically, Calvin first imposes an error (with probability $1/2$) on every entry $x_i$ of the transmitted codeword that differs from that chosen by Calvin $x_i'$, changing $x_i$ to $x_i'$. This operation pushes the transmitted codeword towards the codeword selected by Calvin.
Once Calvin has exhausted his budget of $pn$ errors, he moves to erasures and erases any entry $x_i$ that differs from $x_i'$.
If Calvin's $\pstar n$ budget allows him to erase all such symbols, by symmetrization techniques (e.g., \cite{dey2012improved}) we show that with constant probability Bob is unable to determine whether Alice transmitted her codeword or the one chosen by Calvin, causing a decoding error with probability $1/2$ in this case.
To prove our bound, the remaining budget of Calvin (of errors and erasures) must suffice to push the codeword of Alice half the distance towards that chosen by Calvin.
Using the $q$-ary Plotkin bound~\cite{blake1976introduction} and some additional ideas, one can show that with constant probability the distance between these two codewords on the locations of the push phase is at most  $(1-1/q)(n-b)$, implying that Calvin needs a remaining budget for the last $n-b$ channel uses in which the number of erasures plus twice the number of errors is at least $(1-1/q)(n-b)$.

Roughly speaking, calculations show that for every $\bar{p} \leq p$ there is a corresponding threshold $b$ for which Calvin's budget suffices for the push phase.
However, one would like $b$ to be ``just long enough''.
Setting $b$ to be too small will shorten the babble phase of Clavin and will increase the block length of the push phase and as such will increase the budget needed by Calvin to overcome the potential distance of $(1-1/q)(n-b)$ between his and Alice's codeword. Too long of a babble phase makes Calvin's attack look more similar to the output of a random channel, resulting in a weaker outer bound.
All in all, the threshold $b$ is set to be the minimal value possible that still leaves Calvin with a sufficient ``push'' budget.

Given $p$, $\pstar$, $q$ and $\bar{p}$ the parameter $b$ is set to roughly the value  $\alfq n$ (specified in  Theorem \ref{thm:main:capacity}) which implies that the babble phase behaves like a $q$-ary symmetric channel with error parameter $\bar{p}/\alfq$ (recall that in the babble phase Calvin is changing $\bar{p}n$ randomly chosen locations out of the $b$ locations in the phase). Hence, the upper bound obtained in this case is the rate of the corresponding $q$-ary symmetric channel with block length $b=\alfq n$, which is exactly that stated in the term of Theorem~\ref{thm:main:capacity}.

As we will see shortly in our achievability scheme, setting the rate just below the upper bound (for optimal $\bar{p}$) allows us to overcome Calvin's pushing capabilities and as such allows successful communication, implying a tight characterization of the capacity for our online model. 

\subsubsection{Achievability}

In our codes the encoder Alice uses internal randomness (not known to Bob or Calvin) in the choice of the transmitted codeword, designed to allow a high probability of successful communication no matter which message Alice is sending to Bob. 
\zx{We use ``chunked random codes'' described shortly. That is, we pick our codes uniformly at random from a random ensemble specified in Section~\ref{sec:model},}
and prove that w.h.p. over the code distribution a code chosen at random allows reliable communication.
The decoder involves two major phases: a list decoding phase in which the decoder obtains a short list of messages that include the one transmitted; and a unique decoding phase in which the list is reduced to a single message. Roughly, Bob in his decoding process divides the received word into two parts -- all symbols received up to a given time $\tstar$, and all symbols received afterwards. The list decoding is done using the first part of the received word, and the process of unique decoding from the list is done using the second part. 

Consider first the special case in which there are erasures only.
In this case, given the parameter $\pstar$ (that specifies the fraction of symbols that can be erased by the adversary) and the received word, the decoder Bob can pin-point the value of $\tstar$ that will allow successful decoding. Specifically, for any adversarial behavior, we show the existence of a value $\tstar$ that on one hand allows Bob to obtain a small list of messages from the first part of the received word; and on the other guarantees that the fraction of symbols erased by the adversary in the second part of the received word cannot suffice to confuse Bob between any two messages in the list he holds. 
Notice the duality between the parameter $b$ of our upper bound and the parameter $\tstar$ here.
For our upper bound, we show that above rate $C$ no matter the code shared by Alice and Bob there exists a threshold $b$ for which Bob cannot uniquely decoding based on the first $b$ received symbols and Calvin has a sufficient remaining budget to cause a decoding error in the remaining $n-b$ symbols.
In our lower bound, for any rate below $C$ we suggest a coding scheme and show that there exists a threshold $\tstar$ for which Bob can list decode based on the first $\tstar$ received symbols and that Calvin does not have sufficient budget left to cause a decoding error in the remaining $n-\tstar$ symbols.
As the rate for list decoding (in our lower bound) resembles that of the $q$-ary symmetric channel (in our upper bound) we obtain tight results.

The ability to list decode is obtained using standard probabilistic arguments that take into account the block length $\tstar$ and the number of erasures $\lambda_{\tstar}$ in the first part of the received word.
The ability to uniquely decode from the obtained list involves a more delicate analysis which uses the stochastic nature of our encoding and the causality constraint of Calvin. In particular,  we use the fact that the secret symbols used in the encoding of the first part of the codeword (up to position $\tstar$) are independent of those used for the second part. This independence is useful in separating the two decoding phases 
in the sense that the casual adversary at time $\tstar$ is acting with no knowledge whatsoever on the secret symbols used by Alice after time $\tstar$. This lack of knowledge sets the stage for the unique decoding phase.

We accommodate different potential values of $\tstar$ by designing a stochastic encoding process in which different parts of the codewords rely on independent secret symbols of Alice. Namely, we divide the coding process into {\em chunks}. Each chunk is a random stochastic code of length $\chunk$ for a small parameter $\echunk$ that uses independent randomness from Alice. The final code of Alice is a concatenation of all its chunks. Setting $\echunk$ small enough allows enough flexibility to manage any possible value $\tstar$ chosen by Bob's decoder.

\begin{figure}[t!]
			\centering
			\includegraphics[width=1\linewidth]{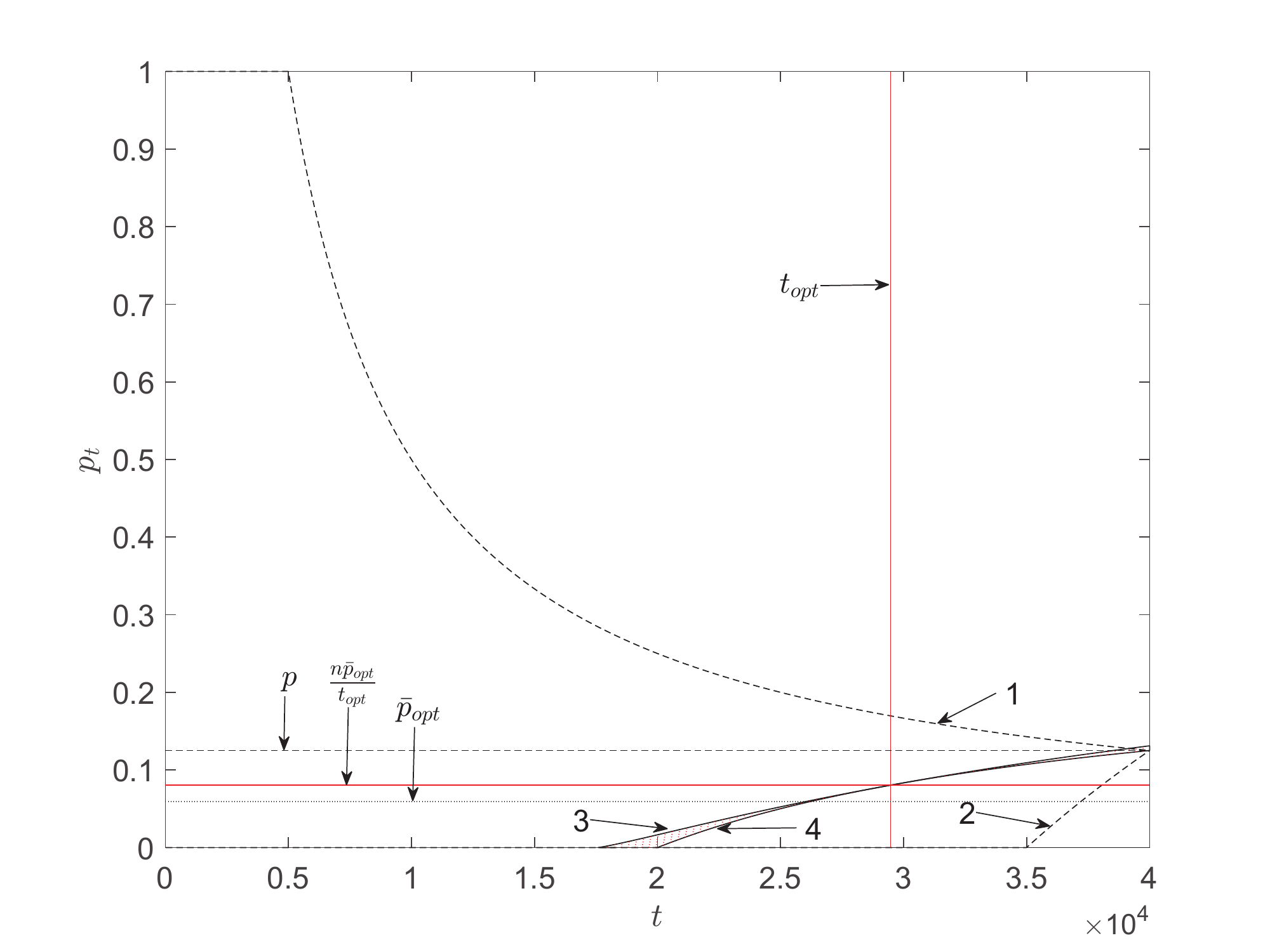}
			\caption{The range for trajectory $\phatt$ (shaded) as a function of $t$ for $q=2$, $p = 1/8$, $p^* = 0$. Our bounds are analytical, however the plot was made numerically using $n=40,000$. Curves 1 and 2 are extremal curves for Calvin's true corruption fraction $p_t$. Curves 3 and 4 bound the region for $\phatt$. Horizontal lines $p$ and $\bar{p}_{opt} $ (optimal $\bar{p}$ from upper bound) are given as references. If Calvin were to follow the attack given in our upper bound proof, then $p_t=\frac{n\bar{p}_{opt}}{t_{opt}}$ (red horizontal line) and in our decoding scheme $\phatt=p_t$ at point $t_{opt}$ (red vertical line). For other values of $p_t$, the location in which $\phatt=p_t$ will differ.}

			\label{fig:traj}
\end{figure}

The encoding and decoding process for the channel in the presence of both errors and erasures follow the same line of analysis as specified above for the erasure only case, but with one major and significant difference. Bob does not know which symbols in the transmitted codeword were in error, and thus by studying the received word, Bob is not able to identify a location $\tstar$ with the desired properties.
To overcome this difficulty, we design an iterative 
decoding process in which Bob starts with a small value of $\mt$ and performs an attempt to decode.
As before the decoding process first list decodes using the first part of the received word and then uniquely decodes. 

The list decoding is done according to a certain ``guessed'' value $\phatt$ for the fraction of symbol errors in the first part of the received word. Here, $\phatt$ is a carefully designed function of $\mt$  (also referred to as a ``trajectory'') that is fixed and known to all parties involved in the communication. 
The trajectory $\phatt$ is chosen in a way that guarantees successful decoding for any location $t$ for which $\phatt$ equals the fraction of symbols $p_t$ actually changed by Calvin up to location $t$ (with respect to unerased positions).
Specifically, $\phatt$ guarantees that Bob is able to obtain a small list of messages by list decoding up to position $t$ and to uniquely decode from this list as the remaining corruption power of Calvin is limited.
Analyzing these conditions gives a range of possible trajectories $\phatt$ depicted in Figure~\ref{fig:traj}.
If $\lat$ denotes the number of erasures Bob receives after $t$ channel uses, then for $\mt-\lat < n \left(1-\frac{2q}{q-1}\ppp-\frac{q}{q-1}\pstar\right)$, we set $\phatt=0$; otherwise we set $\phatt = \ppp+\frac{\pstar}{2}-\left(\frac{q-1}{2q}-p-\frac{\pstar}{2}\right)\left(\frac{n}{\mt-\lat}-1\right)$.  The value of $\phatt$ is $0$ for all $t-\lambda_t$ up to $n \left(1-\frac{2q}{q-1}\ppp-\frac{q}{q-1}\pstar\right)$ and then it grows up to $\frac{\ppp}{1-\frac{q}{q-1}\pstar}$ as $\mt-\lat$ increases to  $n\left(1-\frac{q}{q-1}\pstar\right)$ \zx{(note that since $\lat$ is bounded from above by $n\pstar$, therefore as $\mt$ ranges from $0$ to $n$, the quantity $\mt-\lat$ always takes all possible integer values from $0$ to (at least) $n(1-\pstar)$)}. 

Now that we have $\phatt$, we show that the iterative decoding of Bob is successful at threshold location $t$ if indeed $\hat{p}_t =p_t$, otherwise, we show that the unique decoding phase will fail in the sense that Bob will not receive any message from the decoding process. Identifying a failure in the decoding process, Bob increases $\mt$ and repeats the decoding attempt. The crux of our analysis lies in our proof that eventually, no matter what the behavior of Calvin is, there will be a value of $\mt$, denoted $\tstar$, for which $\phat_{\tstar}$ is (approximately) $p_{\tstar}$ and the decoding succeeds.
Establishing the existence of the trajectory $\phatt$ as discussed above and proving that at some point it must be close to $p_t$ is a central part of our proof.

\subsection{Structure}

In Section~\ref{sec:model} we formally present the channel model, the encoder, and the decoding process. In addition, we present a careful description of the adversarial behavior. Section~\ref{sec:code-analysis} then presents an overview of our code analysis, and the proof of the achievability of Theorem~\ref{thm:main:capacity}. 
Due to space limitations, all the technical claims and their proofs appear in the Appendix.

\section{Model}
\label{sec:model}

\paragraph{Channel Model:}
For any positive integer $i$, let $\left[ i\right] $ denote the set $\left\lbrace 1,2,\cdots,i\right\rbrace $. For a transmission duration of $n$ symbols, a $q$-ary causal adversarial error-erasure channel can be characterized by two triples $(q,p,\pstar)$ and $\left( \mathcal{X}^n,\textsf{Adv},\mathcal{Y}^n\right) $. Here, $p$ and $\pstar$ are the fractions of symbol errors and symbol erasures that Calvin can impose on a codeword, $\mathcal{X}=\left\lbrace 0,1,\cdots,q-1\right\rbrace $ and $\mathcal{Y}=\left\lbrace 0,1,\cdots,q-1\right\rbrace \cup\left\lbrace\Lambda\right\rbrace$ are the input and output alphabet of the channel, and $\textsf{Adv}=\left\lbrace \text{Adv}^i|i\in\left[ n\right] \right\rbrace $ is a sequence of mappings that represents the adversarial behavior in each time step.
More precisely, each map $\text{Adv}^i:\mathcal{X}^i\times\mathcal{Y}^{i-1}\to\mathcal{Y}$ is a function that, at the time of transmitting the $i$-th symbol, maps the sequence of channel inputs up to time $i$, $\left( x_1,x_2,\cdots,x_i\right) \in\mathcal{X}^i$, together with the sequence of all previous channel outputs up to time $i-1$, $\left( y_1,y_2,\cdots,y_{i-1}\right) \in\mathcal{Y}^{i-1}$, to an output symbol $y_i\in\mathcal{Y}$.
The functions  $\text{Adv}^i$ must satisfy the adversarial power constraint, namely that at no point in time does the total number of errors and erasures exceed $pn$ and $\pstar n$, respectively.

\paragraph{Random code distribution:}
We now define a distribution over codes. In our proof, we use this distribution to claim the existence of a fixed code that allows reliable communication between Alice and Bob over the channel model.
In our code construction $R$ denotes the code \textit{rate}, $S$ the {\em private secret rate} of the encoder (to be defined explicitly shortly), and $\echunk$ a ``quantization'' parameter  (specified below).


Let $\msg=\msgspace$ denote Alice's message set and $\secr=\keyspace$ be the set of private random secrets available only to Alice. The encoder randomness $\secr$ is neither shared with the receiver nor the adversary. Let $\Phi$ be the uniform distribution over stochastic codes $\msg \times \secr \to \mathcal{X}^{\chunk}$. Let $\mathcal{C}_1,\mathcal{C}_2,\cdots,\mathcal{C}_{\chunknum}$ be stochastic codes, which are i.i.d. according to the probability distribution $\Phi$. Specifically, $\forall i\in\left[ \chunknum\right] $, the corresponding stochastic code is a map $\mathcal{C}_i:\msg \times \secr \to \mathcal{X}^{\chunk}$ chosen from the distribution $\Phi$.  

\paragraph{Encoder:}
Given a message $m\in\msg$ and $\chunknum$ secrets, $s_1,s_2,\cdots,s_{\chunknum}$ each in $\secr$, a codeword of length $n$ with respect to the message $m$ and the $\chunknum$ secrets is defined to be the concatenation of $\chunknum$ {\em chunks} of sub-codewords,
\begin{align}
	\megacodw[\chunknum] \label{code-design}
\end{align}
where $\mathcal{C}_{i}(m,s_{i})$ is the $i$-th \textit{sub-codeword} in the entire codeword, and $\circ$ denotes the concatenation between two chunks of sub-codewords. To distinguish the concatenated code $\mcode$ from the code for a chunk, we will call $\mathcal{C}_1,\mathcal{C}_2,\cdots,\mathcal{C}_{\chunknum}$ \textit{sub-codes} hereafter. Our code analysis then focuses on two different parts of the entire code, defined as follows.

\begin{definition}
	\label{def:left-mega-code}
	Let a code $\mcode$ of block-length $n$ consist of $\chunknum$ sub-codes, i.e., $\mcode=\megacode[\chunknum]$. Let $\setchend=\choiceschunk$ and $t\in\setchend$. A code prefix of $\mcode$ with respect to $t$ is the concatenation of the first $\frac{t}{\chunk}$ sub-codes of $\mcode$.
\end{definition}

\begin{definition}
	\label{def:right-mega-code}
	Let a code $\mcode$ of block-length $n$ consist of $\chunknum$ sub-codes, i.e., $\mcode=\megacode[\chunknum]$. Let $\setchend=\choiceschunk$ and $t\in\setchend$. A code suffix of $\mcode$ with respect to $t$ is the concatenation of the last $\frac{1}{\echunk}-\frac{t}{\chunk}$ sub-codes of $\mcode$.
\end{definition}

In our analysis, it is convenient to describe the encoding scheme of Alice in a causal manner. Namely, we will assume that the secret value $s_i$ corresponding to the encoding of the $i$-th chunk is chosen by Alice immediately before the $i$-th chunk is to be transmitted and no sooner.

As mentioned above, we show that with positive probability, the code $\mathcal{C}$ chosen at random based on the distribution above has certain properties that allow reliable communication over our channel model. 

\paragraph{Decoding process:}
The decoding process of Bob is done in an iterative manner. Specifically, upon receiving the entire codeword with errors and erasures, for some fixed $\erate>0$, Bob identifies the smallest value of $\mt-\lat \geq n \left(1-\frac{2q}{q-1}\ppp-\frac{q}{q-1}\pstar-\frac{\erate^2}{4}\right)$ corresponding to the (end) location of a chunk, and attempts to correctly decode the transmitted message $m$ based on the codeword prefix and suffix with respect to position $t$. 
The decoding process is terminated if a message is decoded by Bob, otherwise the value of $t$ is increased by $\chunk$ (the chunk size) and Bob attempts to decode again. This process continues until $t$ reaches (approximately) the end of the codeword. If no decodings succeeds until then, a decoder error is declared. 

Each attempt of decoding can be divided into two phases. First, at each position $\mt$, Bob chooses an estimate $\phatt$ for the fraction of errors (with respect to the unerased positions) used by Calvin in the codeword prefix up to $\mt=k\chunk$. In our proof to come, we show that $\phatt$ satisfies two important conditions, the \textit{list-decoding condition} and the \textit{energy bounding condition} (see Claim~\ref{claim:list-energy}). The list-decoding condition allows Bob to decode the codeword prefix $\megacodw[k]$ through a list decoder with list size $\mlists$. As we will show, the list size $\mlists$ consists of at most $\blistord$ messages. 
So at this phase Bob obtains a list $\mlist$ of $\mlists$ messages.
If it is the case that $\phatt$ equals the true fraction of symbol errors $\pref{p}$ (with respect to the unerased positions) up to $\mt$, then it holds that the transmitted message is in $\mlist$. 

Next, for the second phase, the energy bounding condition states that, if $\phatt$ equals $\pref{p}$, there are no more than \zx{$\left(\frac{q-1}{2q}-\frac{\erate^2}{9q^{2}}\right) \left(n-t-n\pstar+\lat\right)-\frac{n\pstar}{2q}$} symbol errors in the codeword suffix with respect to position $\mt$. Therefore, as we will show, Bob can use a natural \textit{consistency} decoder (defined below) to determine whether to stop or continue the decoding process. More precisely, the decoding process continues if the consistency decoder fails to return a message and stops if a message $\hat{m}$ is decoded from the messages in $\mlist$. The decoder also stops when $\mt-\lat$ has reached size $n-\frac{q}{q-1}n\pstar-\chunk$, where $\lat$ is the number of erasures up to position $\mt$.

\begin{definition}
	Let $\epsilon > 0$. Let \zx{$\mathbf{y}_{t},\mathbf{y}^{\prime}_{t}\in\mathcal{Y}^{n-t}$} be two word suffixes with respect to position $t$. The word suffix \zx{$\mathbf{y}_{t}$} is \textit{consistent} with the word suffix \zx{$\mathbf{y}^{\prime}_{t}$} if and only if the \zx{fraction of the unerased positions in which $\mathbf{y}_{t}$ does not agree with $\mathbf{y}^{\prime}_{t}$ is no more than $\frac{q-1}{2q}-\frac{\erate^2}{9q^{2}}-\frac{n\pstar}{2q(n-\mt-n\pstar+\lat)}$}.
\end{definition}

\begin{definition}
	\label{def:consistency-decoder}
	A consistency decoder applied to a code suffix $\rmegacode[k]$ with respect to position $\mt=k\chunk$ and list $\mlist$ is a decoder that takes the word suffix of a received word $\mathbf{y}^{\prime}$ and returns a unique message $\hat{m}$ in the list $\mlist$, one of whose codeword suffixes is consistent with that of $\mathbf{y}^{\prime}$. If more than one such message exists, then a decoding error is declared.
\end{definition}

Formally, the decoder process of Bob can be described as follows. Essentially, we will use the following definition of $\phatt$ (the estimate to Calvin's error corruption fraction with respect to unerased positions at time $\mt$ used by Bob), which is slightly revised later in Definition~\ref{def:p-hat-t} to be more robust to slight slacknesses that appear in the analysis. Let $p\in\left(0,\frac{q-1}{2q}\right)$, then for $\mt-\lat < n \left(1-\frac{2q}{q-1}\ppp-\frac{q}{q-1}\pstar\right)$, $\phatt=0$; otherwise $\phatt = \ppp+\frac{\pstar}{2}-\left(\frac{q-1}{2q}-p-\frac{\pstar}{2}\right)\left(\frac{n}{\mt-\lat}-1\right)$.  The value of $\phatt$ is $0$ for all $t$ up to $n \left(1-\frac{2q}{q-1}\ppp-\frac{q}{q-1}\pstar\right)$ and then it grows up to $\frac{\ppp}{1-\frac{q}{q-1}\pstar}$ as $\mt-\lat$ increases to $n\left(1-\frac{q}{q-1}\pstar\right)$. For the description below, recall that $\epsilon>0$ is a constant design parameter that can be considered to be arbitrarily small.

\noindent {\bf 1.} Identify the position $t=t_{0}=k_{0}\chunk$ for some integer $k_{0}$, where $t_{0}$ is the smallest integer such that $\mt_0 - \lat[\mt_{0}]\geq n\left(1-\frac{2q}{q-1}\ppp-\frac{q}{q-1}\pstar-\frac{\erate^2}{4}\right)$.\newline
\hspace{0.1in} {\bf 2.} List-decode the code \zx{prefix} $\megacode[k]$ with respect to position $t$ to obtain a list $\mlist$ of messages of size $\mlists$, with the list-decoding radius $(\mt-\lat)\phat_{t}$. 
More precisely, a message $m$ is in the list $\mlist$ if there is a codeword corresponding to $m$ for which its unerased symbols in the codeword prefix with respect to position $t$ is of distance no more than $(\mt-\lat)\phatt$ from the corresponding unerased symbols in the received word prefix.\newline
\noindent {\bf 3.} Verify the codeword suffixes with respect to position $t$ corresponding to messages in the list $\mlist$ through a consistency decoder \zt{that compares symbols in unerased positions}. Specifically, consider the Hamming balls with radius equal to $\left(n-n\pstar-\mt+\lat\right)\left(\frac{q-1}{2q}-\frac{\erate^2}{9q^{2}}\right)\zx{-\frac{n\pstar}{2q}}$ centered at the codeword suffix of each codeword corresponding to the messages in the list $\mlist$. If the corresponding received word suffix is outside all the balls, increase $t$ by $\chunk$ and goto Step (2). If the received word suffix lies in exactly one of the balls, decode to the message $\hat{m}$ corresponding to the center of the ball. If the received word suffix lies in more than one ball a decoding error is declared.

For every message $m$, Bob decodes correctly if his estimate $\hat{m}$ equals $m$. 
That is, Bob decodes correctly if for some $\tstar$, the only codeword suffix of the codewords corresponding to messages in the list $\mlist$ consistent with that of the received word corresponds to the message $m$. We show that this indeed happens w.h.p. over the random secrets $\secr^{\frac{n-\tstar}{\chunk}}$ used by Alice for the codeword suffix with respect to position $\tstar$.
If Bob's estimate $\hat{m}$ is not equal to $m$, Bob is said to make a \textit{decoding error}. The \textit{probability of error} for a message $m$ is defined as the probability over Alice's private secrets $s\in\secr$ that Bob decodes incorrectly. The probability of error for the code $\mcode$ is defined as the maximum of the \zx{probabilities} of error for message $m$ over all messages $m\in\msg$.

A rate $R$ is said to be \textit{achievable} if for every $\xi>0,\beta>0$ and every sufficiently large $n$ there exists a code of block length $n$ that allows Alice to communicate $q^{n(R-\beta)}$ distinct messages to Bob with probability of error at most $\xi$. The supremum over $n$ of all achievable rates is the {\it capacity} $C$ of the channel.

\paragraph{Adversarial behavior:}
The behavior of Calvin is specified by the channel model above. In particular, we are more interested in how Calvin \zx{corrupts} a codeword with errors, which can be characterized by a function $\pref{p}$ defined below which specifies how many errors were ejected by Calvin up-to position $t$ normalized by the number of unerased positions. We refer to $\pref{p}$ as a trajectory, and note that the exact trajectory used by Calvin is not known to the decoder Bob.

\begin{definition}[Calvin's Trajectory $\pref{p}$]
	\label{def:p-t}
	Let a codeword $\mx$ of length $n$ consist of $\chunknum$ chunks of sub-codewords. Let $\setchend=\choiceschunk$ and $\leftcw\in\setchend$. \zt{Let $\ppp_{\mt} \in[0,1]$ be the actual fraction of symbol errors with respect to the unerased positions in the codeword prefix of $\mx$ with respect to position $t$}.
\end{definition}

In our analysis we assume that Calvin has certain capabilities that may be beyond those available to a causal adversary. This is without loss of generality as we are studying lower bounds on the achievable rate in this work.
We assume that the trajectory of $\phatt$ that Bob uses in his decoding process is known to Calvin.
This implies (as we will show) that Calvin knows the position $\tstar$ that Bob eventually stops his decoding process. 
In addition, we assume that the list of messages obtained through Bob's list decoding process can be determined explicitly by Calvin. Moreover, we assume that Calvin knows the message $m$ {\it a priori}. 

At every list-decoding position $\mt=k\chunk$, we stress that the subsequent secrets, namely, $(s_{k+1},s_{k+2},\cdots,s_{\chunknum})$ for the codeword suffix are unknown to Calvin. Indeed, given the causal nature of Alice's encoding, these secrets have not even been chosen by Alice at this point in time. The fact that the secrets are hidden from Calvin implies that $(s_{k+1},s_{k+2},\cdots,s_{\chunknum})$ are completely independent of the list (obtained through Bob's list decoding) $\mlist$ determined by Calvin. This fact is crucial to our analysis.

Also, we strengthen Calvin by allowing him to choose which symbols to corrupt after position $\tstar=\kstar\chunk$ non-causally. Namely, we assume that Calvin chooses his corruption pattern after looking ahead to all the remaining symbols of the transmitted codeword. As we show, no matter how these corruptions are chosen, the codeword suffix has \zx{at most $\left(n-n\pstar-\tstar+\lat[\tstar]\right)\left(\frac{q-1}{2q}-\frac{\erate^2}{9q^{2}}\right) -\frac{n\pstar}{2q}$ symbols in error}. The fact that the distribution of $(s_{\kstar+1},s_{\kstar+2},\cdots,s_{\chunknum})$ is independent from the list $\mlist$ will allow us to show that Bob succeeds in his decoding.

\section{Code Analysis}
\label{sec:code-analysis}
Due to space limitations, the technical details of our proof appear entirely in the Appendix. In what follows, we give a roadmap for our proof, including the major high-level arguments used in the Appendix. 
Throughout, $\epsilon>0$ is a constant design parameter that can be considered to be arbitrarily small.

\paragraph{Existence of trajectory $\phatt$:} Our analysis of Bob's decoding begins with selecting a {\it decoding reference trajectory} $\phatt$ (Definition~\ref{def:p-hat-t}) as a proxy trajectory for Calvin's trajectory $\pref{p}$. Recall that for each $t$, $\pref{p}$ is the fraction of errors (with respect to unerased positions) in the codeword prefix up to $t$, and accordingly, $\phatt$ is the fraction of symbols (with respect to unerased positions) that Bob {\em assumes} are in errors up to position $t$. In general, the trajectories $\phatt$ and $\pref{p}$ are not equal.
We show in Claim~\ref{claim:list-energy}, that for $\mt-\lat \geq n \left(1-\frac{2q}{q-1}\ppp-\frac{q}{q-1}\pstar-\frac{\erate^2}{4}\right)$ the selected decoding reference trajectory $\phatt$ satisfies two important conditions, the \textit{list-decoding condition} \eqref{eq:list-con-c-body} and the \textit{energy bounding condition} \eqref{eq:energy-con-c-body} introduced \zx{below}.
\begin{align}
\label{eq:list-con-c-body}
& (\mt-\lat)\left( 1-H_{q}(\phatt)\right)-\frac{n\erate}{4} \geq nR\\
\label{eq:energy-con-c-body}
& n\ppp-(\mt-\lat)\phatt+\errdist\leq\frac{q-1}{2q}\left(n-n\pstar-\mt+\lat\right)
\end{align}
The list decoding condition guarantees a small list size if decoding is done with radius $(\mt-\lat)\phatt$; and the energy bounding condition restricts the remaining errors that the adversary has for the codeword suffix if Bob's estimate $\phatt$ to $\pref{p}$ is approximately correct.

To prove correctness of our decoding procedure, we must introduce a new trajectory $\tilde{p}_{t}$,
which is closely related to its counterpart $\phatt$ in the sense that $\tilde{p}_{t}$ approximately equals $\phatt$.
but the former is slightly smaller than the latter.
This parameter is introduced to allow robustness in our analysis which absorbs certain slacknesses that are a result of our code construction and analysis technique (e.g., such as the fact that our chunk size $\chunk$ cannot be made too small). We here give our precise definitions, which can be at times better understood intuitively if the reader keeps the above discussion in mind. All our \zx{notation is} given in Table~\ref{tab:q-params}.

\paragraph{Existence of position $\tstar$ for which $\phat_{\tstar} \simeq p_{\tstar}$:}
Next in our analysis we chooses for some integer $k_{0}$ the position $t_{0}=k_{0}\chunk \simeq n \left(1-\frac{2q}{q-1}\ppp-\frac{q}{q-1}\pstar-\frac{\erate^2}{4}\right) + \lat[\mt_0]$ as a {\em benchmarking} position, and separate our analysis into two cases based on whether $p_{t_{0}}$ is greater than $\phat_{t_0}$ or not.
We use the following classification:

\begin{definition}[High Type Trajectory]
	For any trajectory $\pref{p}$ of Calvin, consider the values of $\pref{p}$ and $\phatt$ at position $t=t_{0}$. If $p_{t_{0}}\geq\phat_{t_0}$ then Calvin's trajectory $\pref{p}$ is a \textit{high} type trajectory.
\end{definition}

\begin{definition}[Low Type Trajectory]
	For any trajectory $\pref{p}$ of Calvin, consider the values of $\pref{p}$ and $\phatt$ at position $t=t_{0}$. If $p_{t_{0}}<\phat_{t_0}$ then Calvin's trajectory $\pref{p}$ is a \textit{low} type trajectory.
\end{definition}


For any High Type Trajectory of Calvin, we show in Claim~\ref{claim-above} that $\pref{p}$ always intersects with $\phatt$ at some point $t$ after $t_{0}$ no matter what corruption pattern is chosen by Calvin (i.e., at point $t$, Bob's estimate $\phatt$ is equal to the actual amount of errors $\pref{p}$).
Moreover, by Claim~\ref{claim-ingap} and Claim~\ref{claim:p-tilde-t}, this implies a value $\tstar$ (the chunk end which falls immediately after the intersection point $t$ above) for which it is guaranteed that the remaining error budget of Calvin is low in the sense that the number of errors that Calvin can introduce in the codeword suffix with respect to $\tstar$ is less than $\left(n-n\pstar-\tstar+\lat[\tstar]\right)\left(\frac{q-1}{2q}-\frac{\erate^2}{9q^2}\right)\zx{-\frac{n\pstar}{2q}}$. 
On the other hand, for any Low Type Trajectory of Calvin, we already know that $\pref{p}$ is approximately $\phatt$ at the point $t_{0}$ (they are both nearly 0). Thus we show in Claim~\ref{claim-zero-starting} that setting $\tstar$ to be equal to $t_0$ we are again guaranteed that the remaining error budget of Calvin is low in the sense that the number of errors that Calvin can introduce in the codeword suffix with respect to $\tstar$ is less than $\left(n-n\pstar-\tstar+\lat[\tstar]\right)\left(\frac{q-1}{2q}-\frac{\erate^2}{9q^2}\right)\zx{-\frac{n\pstar}{2q}}$. Formally:

\begin{definition}
	\label{def:t-star}
	Let $\epsilon>0$ and $\echunk=\echval$. Let $\setchend=\choiceschunk$ and $t\in\setchend$. 
	\begin{enumerate}[label=(\roman*)]
		\item if $p_{t_{0}}<\zxx{\phat_{t_0}}$, $\tstar=t_{0}=k_{0}\chunk$. 
		\item if $p_{t_{0}}\geq\zxx{\phat_{t_0}}$, $\tstar$ is the smallest value in $\setchend$ such that $p_{\tstar-\chunk}>\phat_{\tstar-\chunk}$ and $p_{\tstar}\leq\phat_{\tstar}$.
	\end{enumerate}
\end{definition}

\paragraph{Success of Bob's decoding:}
Bob starts decoding at position $t_{0}$ and continues to decode at subsequent chunk ends until a message is returned by the consistency decoder or until Bob reaches the end of the received word. Claim~\ref{thm:b-list-decodable} and Corollary~\ref{coro:b-list-decodable} (via the list decoding condition \eqref{eq:list-con-c}) guarantee that Bob in his first phase of decoding will always obtain a list of \sj{messages} of list size $L=\blistord$ from the list decoder no matter what position $t$ is currently being considered. 
The analysis in Claim~\ref{thm:b-list-decodable} and Corollary~\ref{coro:b-list-decodable} and in the claims to come is w.h.p. over our random code construction. 
Moreover, for any $t$, the energy bounding condition \eqref{eq:energy-con-c} implies that, in the case of $\pref{p} \simeq \phatt$, the 
unused errors left for Calvin are less than a $\frac{q-1}{2q}-\frac{\erate^2}{9q^2}\zx{-\frac{n\pstar}{2q(n-\mt-n\pstar+\lat)}}$ fraction of the remaining part of unerased symbols of the codeword.

We start by studying the case in which the current iteration of Bob satisfies $t=\tstar$ (which implies that $\pref{p} \simeq \phatt$).
In Claim~\ref{claim:good-1}, Claim~\ref{claim:good-2}, and Claim~\ref{claim:good-3} we show that if $t=\tstar$ Calvin's remaining error budget is not sufficient to mislead the consistency decoder, and will allow unique decoding from the list of messages Bob holds.
Namely, we show that with high probability over the secret random symbols of Alice used in the encoding process, our code design guarantees that the only message in our list that is consistent with the transmitted codeword is the one transmitted by Alice.

More precisely, consider the consistency checking phase of Bob in the iteration in which $t=\tstar$. In this iteration we know (via the energy bounding condition \eqref{eq:energy-con-c}) that the number of unused errors of Calvin is less than a $\frac{q-1}{2q}-\frac{\erate^2}{9q^2}\zx{-\frac{n\pstar}{2q(n-\mt-n\pstar+\lat)}}$ fraction of the remaining part of the unerased symbols of the codeword. At this point in time, Bob holds a small list of messages $\mlist$ that has been (implicitly) determined by Calvin, and via the consistency decoder wishes to find the unique message $m$ in the list that was transmitted. For any transmitted message $m$, as the list is small, we can guarantee that with high probability over our code design most of the codeword suffixes corresponding to $m$ are roughly of distance $\frac{(n-t)(q-1)}{q}$ from any codeword suffix of any other message in the list $\mlist$, which in turn implies, given the bound on Calvin's remaining error budget, that decoding will succeed. However, this analysis is misleading as one must overcome the adversarial choice of $\mlist$ in establishing correct decoding. (We note that a \zx{na\"ive} use of the union bound does not suffice to overcome all potential lists $\mlist$.)

For successful decoding regardless of Calvin's adversarial behavior, we use the randomness in Alice's stochastic encoding (not known {\it a priori} to Calvin) and the fact that Calvin is causal. Recall that every message $m$ can be encoded into several codewords based on the randomness of Alice. Let $s_{left}$ and $s_{right}$ be the collection of Alice's random symbols used up to and after position $\tstar$ respectively. When Calvin (perhaps partially) determines the list $\mlist$ we may assume that he has full knowledge of $s_{left}$.
However by his causal nature he has no knowledge regarding $s_{right}$. 
As the list $\mlist$ is obtained at position $\tstar$ by Bob, we may now 
take advantage of the fact that it is independent of the randomness $s_{right}$ used by Alice. 
Specifically, instead of considering a single codeword in our analysis that corresponds to $m$ we consider the family of codewords that on one hand all share a specific $s_{left}$ (which corresponds to Calvin's view up to position $\tstar$) but have different $s_{right}$. From Calvin's perspective at position $\tstar$, all codewords in this family are equivalent and completely match his view so far. Using a family of codewords that are independent of $\mlist$ in our analysis, and allowing the decoding to fail on a small fraction of them, enables us to amplify the success rate of our decoding procedure to the extent that it can be used in the needed union bound.
Our full analysis is given in Claim~\ref{claim:good-1}, Claim~\ref{claim:good-2}, and Claim~\ref{claim:good-3}.

We now address the case $t \ne \tstar$ in Claim~\ref{claim:p-tilde-t}.
In this case, by previous discussions, it holds that we are in a 
High Type Trajectory of Calvin and that $\pref{p} > \phatt > \tilde{p}_{t}$.
When $t \ne \tstar$ we show that the decoding process of Bob will not return any codewords at all (as all messages in the list will fail the consistency test). In this case, we continue with the next value of $t$ (the next chunk end). 

We summarize all the properties of our code in Claim~\ref{claim:prob-good-exist}. With those properties established, through Bob's iterative decoder we show in Claim~\ref{claim:prob-decode} that Bob is able to correctly decode the transmitted message $m$ w.h.p. over the randomness of Alice. Finally, in Theorem~\ref{thm:capacity} we show that the channel capacity $C$ claimed is indeed achievable. We depict the flow of our claims, corollaries and theorems for the proof of achievability in Figure~\ref{fig:proofs-organization-achievability}.

{{\noindent {\bf Remark:}} The scenario wherein Calvin has $n\epsilon$ lookahead can also be handled via the codes above. Roughly, if we back off in our rate by $\epsilon$ the trajectory $\phatt$ gets shifted to the left by $n\epsilon$. We then ``sacrifice'' $n\epsilon$ symbols to Calvin by demanding that a more stringent energy-bounding condition be satisfied, in which the block length of the second part (succeeding $\tstar$) is reduced by $n\epsilon$. With these tweaks, the remainder of the analysis of the $n\epsilon$-lookahead codes is identical to that of the causal codes discussed above.}

%
%

		\bibliographystyle{unsrt}
		\bibliography{./refs}

\begin{appendices}
		\section{Converse}
		\label{sec:upper-bound}
		
		\begin{figure}[p]
			\centering
			\includegraphics[scale=0.8]{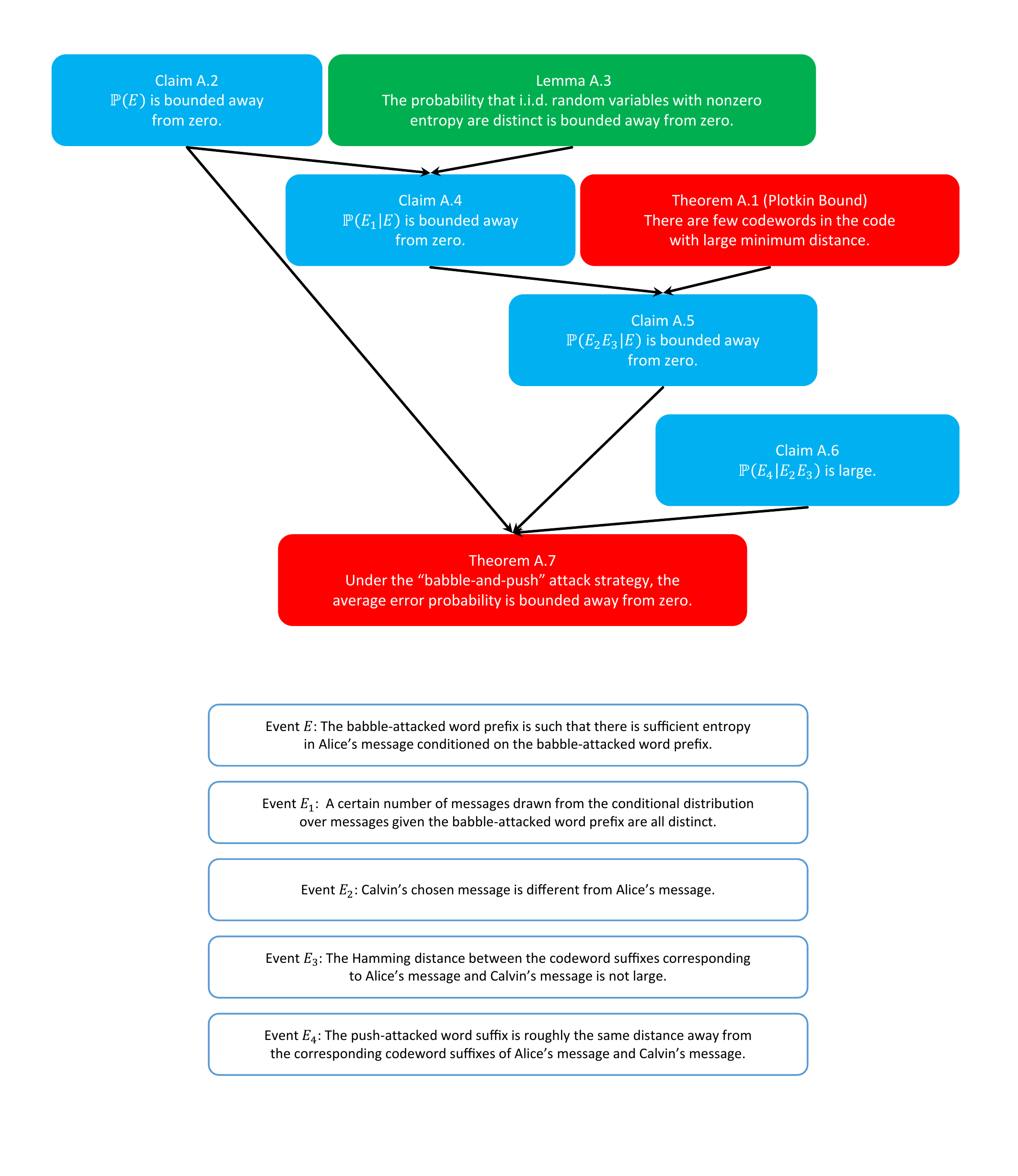}
			\caption{Organization of our claims and theorems for the converse}
			\label{fig:proofs-organization-converse}
		\end{figure}
		
		We start by summarizing several definitions and claims. The detailed presentations of the claims are followed by the summary. We depict the flow of our claims and theorems in Figure~\ref{fig:proofs-organization-converse}.
		\begin{enumerate}[label=\textbf{\arabic*}.]
			\item {\bf Summary of Event Definitions}
			\begin{itemize}
				\item Event~$E$: The babble-attacked word prefix is such that there is sufficient entropy \zx{in} Alice's message (i.e., the transmitted message) conditioned on the babble-attacked word prefix.
				\item Event~$E_{1}$: A certain number of messages drawn from the conditional distribution \zx{over} messages given the babble-attacked word prefix are all distinct.
				\item Event~$E_{2}$: Calvin's chosen message is different from Alice's message.
				\item Event~$E_{3}$: The Hamming distance between the codeword suffixes (with respect to the pushing phase of the attack) corresponding to Alice's message and Calvin's message is not large.
				\item Event~$E_{4}$: The resulting word suffix (with respect to the pushing phase of the attack) is roughly the same distance away from the codeword suffixes (with respect to the pushing phase of the attack) corresponding to Alice's message and Calvin's message.
			\end{itemize}
			\item {\bf Summary of Claims and Theorems}
			\begin{itemize}
				\item Theorem~\ref{thm:plotkin}: There are few codewords in the code with large minimum distance.			
				\item Claim~\ref{cl:sufficient-h}: The probability that $E$ happens is \zx{bounded away from zero}.
				\item Lemma~\ref{le:distinct-rvs}: The probability that i.i.d. random variables with nonzero entropy are distinct is \zx{bounded away from zero}. 
				\item Claim~\ref{cl:e1}: The probability that $E_{1}|E$ happens is \zx{bounded away from zero}.
				\item Claim~\ref{cl:distinct-m-hmd}: The probability that $E_{2}E_{3}|E$ happens is \zx{bounded away from zero}.
				\item Claim~\ref{cl:out-of-budget}: The  probability that $E_{4}|E_{2}E_{3}$ happens is large.
				\item Theorem~\ref{thm:converse}: Under the \zx{``babble-and-push''} attack strategy, the average error probability is \zx{bounded away from zero}.
			\end{itemize}
		\end{enumerate}
		
		Let $q\geq 2$. Let $\ppp\in\left(\intverr\right)$ be the fraction of symbol errors and $\pstar\in\left(\intvera\right)$ be the fraction of symbol erasures. Let $\alfq=\alfexpq$.
		
		In the following, unless otherwise specified, $H\left(\mathbf{X}\right)$ refers to {\em source entropy for symbols} (or $q$-ary entropy), which is obtained through normalizing the standard binary entropy by a factor of $\log{q}$, and $H_{q}\left(x\right)$ refers to {\em the $q$-ary entropy function}, namely, $H_{q}\left(x\right)=x\log_{q}{(q-1)}-x\log_{q}{x}-(1-x)\log_{q}{(1-x)}$.
		
		\subsection*{``Babble-and-push'' Attack}
		\begin{enumerate}
			\item ``Babble'': Let $\blen=n\left(\alfq+\ebab\right)$ be the position in the transmitted codeword, up to which Calvin adopts a ``babble'' strategy. Calvin chooses a random subset $\Gamma$ of $n\pbar$ indices uniformly from the set of all $n\pbar$-sized subset of $\left[\blen\right]$. For any $i\in\Gamma$, Calvin changes the symbol $x_{i}$. More precisely, $y_{i}$ is chosen by Calvin uniformly from $\left\lbrace0,1,\cdots,q-1\right\rbrace\setminus\left\lbrace x_{i}\right\rbrace$.
			
			\item ``Push'': Let $\xbab$ be the first $b$ symbols transmitted by Alice and $\ybab$ be the first $b$ symbols resulting from Calvin's ``babble'' attack, namely, $\xbab=(x_{1},x_{2},\cdots,x_{b})$ and $\ybab=(y_{1},y_{2},\cdots,y_{b})$. Calvin constructs the set of $(m,\mathbf{s})$ pairs that have encodings $\mathcal{C}(m,\mathbf{s})$ that are close to $\ybab$. Specifically, the set constructed by Calvin is
			\begin{align}
			B_{\ybab}=\left\lbrace(m,\mathbf{s})\colon\hmd{\ybab}{\mathcal{C}_{b}(m,\mathbf{s})}=n\pbar\right\rbrace
			\end{align}
			where $\mathcal{C}_{b}(m,\mathbf{s})$ is the first $b$ symbols of $\mathcal{C}(m,\mathbf{s})$. Next, Calvin chooses an element $(\mpm,\mathbf{s}^{\prime})\in B_{\ybab}$ uniformly at random and considers the corresponding encoding $\mathcal{C}\left(\mpm,\mathbf{s}^{\prime}\right)=\mathbf{x}^{\prime}=\left(x^{\prime}_{1},x^{\prime}_{2},\cdots,x^{\prime}_{n}\right)$. For $i>b$, if $x_{i}\neq x^{\prime}_{i}$, Calvin sets $y_{i}=x^{\prime}_{i}$ with probability half until $i=n$ or Calvin uses up $np$ errors. If Calvin uses up $np$ errors but $i<n$, then Calvin erases the subsequent symbols $x_{i}$ whenever $x_{i}\neq x^{\prime}_{i}$ until $i=n$ or Calvin uses up $n\pstar$ erasures.
		\end{enumerate}
		
		\begin{theorem}[$q$-ary Plotkin Bound\cite{blake1976introduction}]
			\label{thm:plotkin}
			There are at most $\frac{qd_{min}}{qd_{min}-(q-1)n}$ codewords in any $q$-ary code of block length $n$ with minimum distance $d_{min}>\left(1-\frac{1}{q}\right)n$.
		\end{theorem}
		
		Let $\mbfu$ be the random variable corresponding to Alice's input message, $\mbfx$ be the random variable corresponding to Alice's input codeword, and $\mbfy$ the random variable corresponding to the output of the channel. Let $\mbfx_{b}$ and $\mbfy_{b}$ be the random variables corresponding to $\xbab$ and $\ybab$, respectively.
		Let $E=\left\lbrace\mbfy_{b}\in \left\lbrace \ybab \colon  H\left(\mbfu|\mbfy_{b}=\ybab\right) \geq \frac{n\epsilon}{4}\right\rbrace \right\rbrace$. 
		
		\begin{claim}
			\label{cl:sufficient-h}
			Let \zt{$b=n\left(\alfq+\ebab\right)$}. Then for the ``babble-and-push'' attack, we have
			\begin{align}
			\mathbb{P}\left[E\right] \geq \frac{\epsilon}{4}.
			\end{align}
		\end{claim}
		
		\begin{proof}
			\zx{Considering} the entropy $H\left(\mbfu|\mbfy_{b}\right)$, we have
			\begin{align}
			H\left(\mbfu|\mbfy_{b}\right) & = H\left(\mbfu\right) - I\left(\mbfu;\mbfy_{b}\right)
			\nonumber
			\\
			& \geq H\left(\mbfu\right) - I\left(\mbfx_{b};\mbfy_{b}\right) 
			\label{eq:data-proc-ineq}
			\\
			& \geq H\left(\mbfu\right) - b\left(1-H_{q}\left(\frac{n\pbar}{b}\right)\right) 
			\nonumber
			\\
			& = H\left(\mbfu\right) - n\left(\alfq+\ebab\right)\left(1-H_{q}\left(\frac{\pbar}{\alfq+\ebab}\right)\right) 
			\label{eq:sub-b}
			\\
			& \geq n\left(\capq[]+\epsilon\right)- n\left(\alfq+\ebab\right)\left(1-H_{q}\left(\frac{\pbar}{\alfq+\ebab}\right)\right) 
			\label{eq:rate-assumed}
			\\
			& = \frac{n\epsilon}{2}+n\left(\left(\alfq+\ebab\right)H_{q}\left(\frac{\pbar}{\alfq+\ebab}\right)-\alfq H_{q}\left(\frac{\pbar}{\alfq}\right)\right)
			\nonumber
			\\
			& \geq \frac{n\epsilon}{2}
			\label{eq:suff-entropy}
			\end{align}
			where \eqref{eq:data-proc-ineq} follows by the data-processing inequality, \eqref{eq:sub-b} follows by substituting $b=\zt{n\left(\alfq+\ebab\right)}$, \eqref{eq:rate-assumed} follows by assuming $R=\capq[]+\epsilon$, and \eqref{eq:suff-entropy} follows by the fact that $xH_{q}\left(\frac{\pbar}{x}\right)$ is a monotonic increasing function in variate $x$.
			
			Therefore, the expected value of $H\left(\mbfu|\mbfy_{b}=\ybab\right)$ over $\ybab$ is at least $\frac{n\epsilon}{2}$ and the maximum value of $H\left(\mbfu|\mbfy_{b}=\ybab\right)$ is $nR$. Applying the Markov inequality to the random variable $nR-H\left(\mbfu|\mbfy_{b}=\ybab\right)$, we have
			\begin{align*}
			\mathbb{P}\left[nR-H\left(\mbfu|\mbfy_{b}=\ybab\right)> nR-\frac{n\epsilon}{4}\right] & < \frac{nR-\frac{n\epsilon}{2}}{nR - \frac{n\epsilon}{4}}
			\\
			& = \frac{R-\frac{\epsilon}{2}}{R - \frac{\epsilon}{4}}
			\end{align*}
			Therefore,
			\begin{align}
			\mathbb{P}\left[E\right] & = \mathbb{P}\left[H\left(\mbfu|\mbfy_{b}=\ybab\right) \geq \frac{n\epsilon}{4}\right] 
			\nonumber
			\\
			& \geq 1 - \frac{R-\frac{\epsilon}{2}}{R - \frac{\epsilon}{4}}
			\nonumber
			\\
			& = \frac{\frac{\epsilon}{4}}{R-\frac{\epsilon}{4}}
			\nonumber
			\\
			& \geq \frac{\epsilon}{4}
			\label{eq:prob-e}
			\end{align}
			where \eqref{eq:prob-e} follows by the fact that $R\leq 1$.
		\end{proof}
		
		\begin{lemma}
			\label{le:distinct-rvs}
			Let $V$ be a random variable on a discrete finite set $\mathcal{V}$ with entropy $H\left(V\right)\geq\mu$, and let $V_{1},V_{2},\cdots,V_{k}$ be i.i.d. copies of V. Then
			\begin{align}
			\mathbb{P}[\lbrace V_{1},V_{2},\cdots,V_{k}\rbrace \text{ are all distinct }]\geq\left(\frac{\mu-\zt{\log_{q}{2}}-\log_{q}{k}}{\log_{q}{\lvert \mathcal{V}\rvert}}\right)^{k-1}.
			\label{eq:le-distinct}
			\end{align}
		\end{lemma}
		
		\begin{proof}
			Fix $i\leq k$ and let $A_{i}=\lbrace v_{1},v_{2},\cdots,v_{i}\rbrace$, where $v_{1},v_{2},\cdots,v_{i}\in\mathcal{V}$. Let $W_{i}=\mathbf{1}\left(V_{i+1}\in A_{i}\right)$, where $\mathbf{1}\left(\cdot\right)$ denotes the {\em indicator function}. We write the distribution of $V$ as
			\begin{align*}
			\mathbb{P}\left[V_{i+1}=v\right]=\sum_{j\in\lbrace 0,1\rbrace}\mathbb{P}\left[W_{i}=j\right]\mathbb{P}\left[V_{i+1}=v|W_{i}=j\right]
			\end{align*}
			Then we can bound from above the entropy of $V$ as
			\begin{align*}
			H\left(V_{i+1}\right) & \leq H\left(V_{i+1}|W_{i}\right)+H\left(W_{i}\right)\\
			& = \sum_{j\in\lbrace 0,1\rbrace}\mathbb{P}\left[W_{i}=j\right]H\left(V_{i+1}=v|W_{i}=j\right)+H\left(W_{i}\right)\\
			& \leq \log_{q}{i}+\mathbb{P}\left[W_{i}=0\right]\log_{q}{\lvert\mathcal{V}\rvert}+\zt{\log_{q}{2}}
			\end{align*}
			Since $H\left(V\right)\geq\mu$, we have
			\begin{align*}
			\log_{q}{i}+\mathbb{P}\left[W_{i}=0\right]\log_{q}{\lvert\mathcal{V}\rvert}+\zt{\log_{q}{2}} \geq \mu
			\end{align*} 
			Hence, we have
			\begin{align*}
			\mathbb{P}\left[W_{i}=0\right]
			\geq \frac{\mu-\log_{q}{i}-\zt{\log_{q}{2}}}{\log_{q}{\lvert\mathcal{V}\rvert}}
			\geq
			\frac{\mu-\log_{q}{k}-\zt{\log_{q}{2}}}{\log_{q}{\lvert\mathcal{V}\rvert}}
			\end{align*}
			The event that each $V_{i}$ is distinct is equivalent to the event that for each $i\in\lbrace 2,3,\cdots,k\rbrace$, $V_{i+1}\notin A_{i}$, which implies $W_{i}=0$.
		\end{proof}

		\begin{claim}
			\label{cl:e1}
			Let $\rho_{\mbfu|\ybab}$ be the conditional distribution of $\mbfu$ given $\ybab$ under the ``babble-and-push'' attack. Let $\mbfu_{1},\mbfu_{2},\cdots,\mbfu_{k}$ be $k$ random variables drawn i.i.d. according to $\rho_{\mbfu|\ybab}$. Let $$E_{1}=\lbrace \lbrace \mbfu_{1},\mbfu_{2},\cdots,\mbfu_{k}\rbrace\text{ are all distinct}\rbrace.$$ For large enough $n$, we have
			\begin{align}
			\mathbb{P}\left[E_{1}|E\right] \geq \left(\frac{\epsilon}{5}\right)^{k-1}.
			\end{align}
		\end{claim}
		
		\begin{proof}
			From Claim~\ref{cl:sufficient-h}, given event $E$, we have $H\left(\mbfu|\mbfy_{b}=\ybab\right) \geq \frac{n\epsilon}{4}$.
			From Lemma~\ref{le:distinct-rvs}, setting $V=\mbfu$, $\mu=\frac{n\epsilon}{4}$, and $\lvert\mathcal{V}\rvert\leq q^{n}$, we have
			\begin{align*}
			\mathbb{P}\left[E_{1}|E\right] \geq \left(\frac{\frac{n\epsilon}{4}-\log_{q}{k}-\zt{\log_{q}{2}}}{n}\right)^{k-1}
			\end{align*}
			For large enough $n$, we have
			\begin{align*}
			\frac{\frac{n\epsilon}{4}-\log_{q}{k}-\zt{\log_{q}{2}}}{n}>\frac{\epsilon}{5}
			\end{align*}
			Thus,
			\begin{align*}
			\mathbb{P}\left[E_{1}|E\right] \geq \left(\frac{\epsilon}{5}\right)^{k-1}
			\end{align*}
		\end{proof}
		
		Let $\mbfu^{\prime}$ be the random choice of Calvin's message and $\mbfx^{\prime}$ be the random variable of the codeword corresponding to $\mbfu^{\prime}$. Let $\xpus=(x_{b+1},x_{b+2},\cdots,x_{n})$ be the remaining part of the input codeword in the ``push'' phase and $\mbfx_{p}$ be the corresponding random variable. Similarly, $\xpus^{\prime}=(x^{\prime}_{b+1},x^{\prime}_{b+2},\cdots,x^{\prime}_{n})$ be the part of the codeword chosen by Calvin in the ``push'' phase and $\mbfx^{\prime}_{p}$ be the corresponding random variable.
		
		\zt{Let $\hmd{\cdot}{\cdot}$ denote the Hamming distance function between two vectors.}
		
		\begin{claim}
			\label{cl:distinct-m-hmd}
			Let $$E_{2}=\lbrace\mbfu\neq\mbfu^{\prime}\rbrace$$ $$E_{3}=\left\lbrace \hmd{\mbfx_{p}}{\mbfx^{\prime}_{p}}\leq 2n\left(p-\pbar\right)+n\pstar-\frac{n\epsilon}{8} \right\rbrace.$$ Then for the ``babble-and-push'' attack, we have
			\begin{align*}
			\mathbb{P}\left[E_{2} E_{3}|E\right] \geq \epsilon^{O\left(\frac{1}{\epsilon}\right)}.
			\end{align*}
		\end{claim}
		
		\begin{proof}
			From Claim~\ref{cl:e1}, setting $k=2$, we lower bound the probability that $E_{2}$ holds given $E$ to be
			\begin{align*}
			\mathbb{P}\left[E_{2}|E\right] \geq \frac{\epsilon}{5}
			\end{align*}
			For general $k$, Claim~\ref{cl:e1} shows that the probability that the $k$ messages drawn from the conditional distribution $\rho_{\mbfu|\ybab}$ are all distinct is at least $\left(\frac{\epsilon}{5}\right)^{k-1}$. On the other hand, Plotkin's bound (Theorem~\ref{thm:plotkin}) shows that there do not exist $q$-ary codes of block length $n-b$ and minimum distance $d$ with more than $\frac{qd}{qd-(q-1)(n-b)}$ codewords. 
			
			Let $A=\left\lbrace \left(m_{i},\mathbf{s}_{i}\right) \colon \left(m_{i},\mathbf{s}_{i}\right)\in B_{\ybab}, i\in\left[k\right]  \right\rbrace$ be a set of $k$ mutually independent pairs uniformly from $B_{\ybab}$. Setting $k=\frac{25}{\epsilon}$, Claim~\ref{cl:e1} and Theorem~\ref{thm:plotkin} together imply that with probability at least $\left(\frac{\epsilon}{5}\right)^{k-1}$ there exist codewords $\mathbf{x}$ and $\mathbf{x}^{\prime}$ corresponding to pairs $\left(m,\mathbf{s}\right)$ and $\left(\mpm,\mathbf{s}^{\prime}\right)$ in $B_{\ybab} $with a distance $d$ satisfying
			\begin{align*}
			\frac{25}{\epsilon} \leq \frac{qd}{qd-(q-1)(n-b)}
			\end{align*} 
			Solving for $d$ and using $b=n\left(\alfq+\frac{\epsilon}{2}\right)$, we have
			\begin{align*}
			d & \leq  2n(p-\pbar)\frac{25}{25-\epsilon}+n\pstar\frac{25}{25-\epsilon}-\frac{n\epsilon}{2}\frac{q-1}{q}\frac{25}{25-\epsilon}\\
			& = 2n(p-\pbar)+n\pstar-\frac{n\epsilon}{4}\left(\frac{2(q-1)}{q}\frac{25}{25-\epsilon}-\frac{8(p-\pbar)}{25-\epsilon}-\frac{4\pstar}{25-\epsilon}\right)\\
			& < 2n(p-\pbar)+n\pstar-\frac{n\epsilon}{8}
			\end{align*}
			
			Let $\Delta=2n(p-\pbar)+n\pstar-\frac{n\epsilon}{8}$. Let $\gamma$ be the fraction of pairs in $B_{\ybab}$ that satisfy $E_{2}$ and $E_{3}$. Then the probability over the selection of set $A$ that event $E_{2}$ and $E_{3}$ hold is
			\begin{align}
			\mathbb{P}\left[\bigcup_{A}\left\lbrace\hmd{\mbfx_{i}}{\mbfx_{j}}<\Delta\text{ and }\left\lbrace\mbfu_{i}\neq\mbfu_{j}\right\rbrack\right\rbrace\right] \leq k^{2}\gamma =  \left(\frac{25}{\epsilon}\right)^{2}\gamma
			\label{eq:e2e3-upper}
			\end{align}
			where $\mbfx_{i}$ and $\mbfx_{j}$ are the codewords corresponding to the pairs $\left(m_{i},\mathbf{s}_{i}\right)$ and $\left(m_{j},\mathbf{s}_{j}\right)$ in set $A$, and $\mbfu_{i}$ and $\mbfu_{j}$ are the corresponding message random variables.
			
			However, the probability that $\lbrace \mbfu_{1},\mbfu_{2},\cdots,\mbfu_{\frac{25}{\epsilon}}\rbrace$ are all distinct and that at least one pair of codewords, $\mbfx_{i}$ and $\mbfx_{j}$ has distance less than $\Delta$ is
			\begin{align}
			\mathbb{P}\left[\bigcup_{A}\left\lbrace\hmd{\mbfx_{i}}{\mbfx_{j}}<\Delta\text{ and }\lbrace \mbfu_{1},\mbfu_{2},\cdots,\mbfu_{\frac{25}{\epsilon}}\rbrace\text{ are all distinct}\right\rbrace\right] \geq \left(\frac{\epsilon}{5}\right)^{\frac{25}{\epsilon}}
			\label{eq:e2e3-lower}
			\end{align}
			
			Since the event analyzed in \eqref{eq:e2e3-upper} includes that in \eqref{eq:e2e3-lower}, we have
			\begin{align*}
			\gamma \geq \left(\frac{\epsilon}{25}\right)^{2}\left(\frac{\epsilon}{5}\right)^{\frac{25}{\epsilon}}=\epsilon^{O\left(\frac{1}{\epsilon}\right)}
			\end{align*}
			
			Hence, by the definition of $\gamma$, we have $\mathbb{P}\left[E_{2} E_{3}|E\right] \geq \epsilon^{O\left(\frac{1}{\epsilon}\right)}$.
		\end{proof}
		
		\begin{claim}
			\label{cl:out-of-budget}
			Let $d$ be the Hamming distance between $\mbfx_{p}$ chosen by Alice and $\mbfx_{p} ^{\prime}$ chosen by Calvin. Let $\mbfy_{p}$ be the corresponding part of the word received by Bob resulting from Calvin's ``push'' attack. Let $$E_{4}=\left\lbrace\hmd{\mbfx_{p}}{\mbfy_{p}}\in\left(\frac{d}{2}-\frac{n\epsilon}{16},\frac{d}{2}+\frac{n\epsilon}{16}\right) \right\rbrace.$$ Then for the ``babble-and-push'' attack, we have
			\begin{align*}
			\mathbb{P}\left[E_{4}|E_{2}E_{3}\right]>1-2^{-\Omega\left(n\epsilon^2\right)}.
			\end{align*}
		\end{claim}
		
		\begin{proof}
			Assume that Calvin erases $n\pstar$ symbols in the ``push'' phase. 
			\zx{\footnote{This actually corresponds to Calvin's ``strongest'' attack -- in the babble phase he uses up a fraction of his budget $np$ symbols errors, and now in the push phase he potentially uses up the remainder of his symbol error budget, and also his $n\pstar$ erasure budget.}} 
			Let $d_{c}=d-n\pstar$ be the Hamming distance between $\mbfx_{p}$ and $\mbfx_{p}^{\prime}$ without considering the positions corresponding to erasures. Then, \zt{if there were no constraints on Calvin's error budget}, Calvin would change $\frac{d_{c}}{2}$ locations in expectation. Conditioned on event $E_{2}$ and event $E_{3}$, we have
			\begin{align*}
			\frac{d_{c}}{2} = \frac{d-n\pstar}{2} \leq n\left(p-\pbar\right)-\frac{n\epsilon}{16} 
			\end{align*}
			Assume that $\frac{d_{c}}{2}=n\left(p-\pbar\right)-\frac{n\epsilon}{16}$. In the ``push'' attack, $d_{c}$ out of $\hmd{\mbfx_{p}}{\mbfx_{p}^{\prime}}$ symbols are drawn, and with probability half, Calvin changes the original symbol in $\mbfx_{p}$ to the intended symbol in $\mbfx_{p}^{\prime}$. By Chernoff's bound, the probability that the number of changes of symbols deviates from the expectation $\frac{d_{c}}{2}$ by more than $\frac{n\epsilon}{16}$ is at most $2^{-\Omega\left(n\epsilon^2\right)}$.
		\end{proof}
		
		\begin{theorem}
			\label{thm:converse}
			For any code with stochastic encoding of rate $R=\capq[]+\epsilon$, under the ``babble-and-push'' strategy, the average error probability $\epsbar$ is lower bounded by $\epsilon^{O\left(\frac{1}{\epsilon}\right)}$.
		\end{theorem}
		
		\begin{proof}
			The idea behind the proof is that conditioned on events $E,E_{2},E_{3}$, and $E_{4}$, Calvin can ``symmetrize'' the channel \cite{csiszar1988capacity,dey2013upper}. That is, Calvin can corrupt symbols in a manner so that Bob is unable to distinguish between two possible codewords $\mathbf{x}$ and $\mathbf{x}^{\prime}$ corresponding to two different messages $m$ and $\mpm$. Calvin does this by ensuring (with probability bounded away from zero) that the word $\mathbf{y}$ received by Bob is equally likely to be decoded to be either $\mathbf{x}$ or $\mathbf{x}^{\prime}$ and their corresponding messages $m$ and $\mpm$.
			
			Let $\rho\left(\ybab,m,\mathbf{s},\mpm,\mathbf{s}^{\prime}\right)$ be the joint distribution of the received word $\ybab$ at the end of the ``babble'' phase, Alice's message and randomness $(m,\mathbf{s})$, and Calvin's chosen message and randomness $(\mpm,\mathbf{s}^{\prime})$, under Alice's uniform choice of $(m,\mathbf{s})$ and Calvin's attack. \zt{For each $\mathbf{y}$, let $\rho\left(\mathbf{y}|\ybab,m,\mathbf{s},\mpm,\mathbf{s}^{\prime}\right)$ be the conditional distribution of $\mathbf{y}$ under Calvin’s attack. Let $\mathcal{D}:\mathcal{Y}^{n}\to\msg$ be a probabilistic map, namely, the mapping $\mathcal{D}(\mathbf{y})$ is a random variable taking values from $\msg$.} The error probability can be written as
			\begin{align*}
				\epsbar=
				\sum_{\ybab,m,\mathbf{s},\mpm,\mathbf{s}^{\prime}}\rho\left(\ybab,m,\mathbf{s},\mpm,\mathbf{s}^{\prime}\right)
				\sum_{\ypus}\rho\left(\mathbf{y}|\ybab,m,\mathbf{s},\mpm,\mathbf{s}^{\prime}\right)\mathbb{P}\left[\mathcal{D}(\mathbf{y})\neq m\right]
			\end{align*}
			Let $\mathcal{F}$ be the set of tuples $\left(\ybab,m,\mathbf{s},\mpm,\mathbf{s}^{\prime}\right)$ satisfying events $E,E_{2}$, and $E_{3}$. Claims~\ref{cl:sufficient-h} and \ref{cl:distinct-m-hmd} show that $$\rho\left(\mathcal{F}\right)\geq\frac{\epsilon}{4}\epsilon^{O\left(\frac{1}{\epsilon}\right)}.$$
			Then for $\left(\ybab,m,\mathbf{s},\mpm,\mathbf{s}^{\prime}\right)\in\mathcal{F}$, we have that $m\neq\mpm$ and that $\hmd{\xpus}{\xpus^{\prime}}$ is sufficiently small.
			
			Assuming $E_{4}$ holds, since Calvin change each symbol in $\xpus$ that is different from that in $\xpus^{\prime}$ with probability half, the corresponding part of the received word, $\ypus$, may result from either $\xpus$ or $\xpus^{\prime}$ with equal probability. Thus, the conditional distribution is symmetric, $$\rho\left(\mathbf{y}|\ybab,m,\mathbf{s},\mpm,\mathbf{s}^{\prime}\right)=\rho\left(\mathbf{y}|\ybab,\mpm,\mathbf{s}^{\prime},m,\mathbf{s}\right).$$ Then, by Claim~\ref{cl:out-of-budget}, for $\left(\ybab,m,\mathbf{s},\mpm,\mathbf{s}^{\prime}\right)\in\mathcal{F}$, we have
			\begin{align*}
			\zt{\sum_{\ypus}\rho\left(\ypus|\ybab,m,\mathbf{s},\mpm,\mathbf{s}^{\prime}\right)
				\geq 1-2^{-\Omega\left(n\epsilon^2\right)}.}
			\end{align*}
			
			Returning to the overall error probability, let $\rho\left(\ybab\right)$ be the unconditional probability of Bob receiving $\ybab$ in the ``babble'' phase, where the probability is over Alice's uniform choice of $(m,\mathbf{s})$ and Calvin's ``babble'' attack. Since the {\em a posteriori} distributions of $(m,\mathbf{s})$ and $(\mpm,\mathbf{s}^{\prime})$ given $\ybab$ are independent and both uniform in $B_{\ybab}$, the joint distribution can be written as
			\begin{align*}
			\rho\left(\ybab,m,\mathbf{s},\mpm,\mathbf{s}^{\prime}\right)=\rho\left(\ybab\right)\frac{1}{\lvert B_{\ybab}\rvert^{2}}=\rho\left(\ybab,\mpm,\mathbf{s}^{\prime},m,\mathbf{s}\right).
			\end{align*}
			Therefore, we have $\rho\left(\ypus|\ybab,m,\mathbf{s},\mpm,\mathbf{s}^{\prime}\right)=\rho\left(\ypus|\ybab,\mpm,\mathbf{s}^{\prime},m,\mathbf{s}\right)$.
			Hence,
			\begin{align*}
			\lefteqn{2\epsbar \geq \sum_{\mathcal{F}}\rho\left(\ybab,m,\mathbf{s},\mpm,\mathbf{s}^{\prime}\right)\cdot}\\ & & &\left( \sum_{\ypus}\rho\left(\ypus|\ybab,m,\mathbf{s},\mpm,\mathbf{s}^{\prime}\right)\mathbb{P}\left[\mathcal{D}\left(\ybab,\ypus\right)\neq m\right]
			+\sum_{\ypus}\rho\left(\ypus|\ybab,\mpm,\mathbf{s}^{\prime},m,\mathbf{s}\right)\mathbb{P}\left[\mathcal{D}\left(\ybab,\ypus\right)\neq \mpm\right] \right)\\
			&  & \geq & \sum_{\mathcal{F}}\rho\left(\ybab,m,\mathbf{s},\mpm,\mathbf{s}^{\prime}\right) \sum_{\ypus}\rho\left(\ypus|\ybab,m,\mathbf{s},\mpm,\mathbf{s}^{\prime}\right) \left(\mathbb{P}\left[\mathcal{D}\left(\ybab,\ypus\right)\neq m\right] + \mathbb{P}\left[\mathcal{D}\left(\ybab,\ypus\right)\neq \mpm\right] \right)\\
			& & \geq & \sum_{\mathcal{F}}\rho\left(\ybab,m,\mathbf{s},\mpm,\mathbf{s}^{\prime}\right) \sum_{\ypus}\rho\left(\ypus|\ybab,m,\mathbf{s},\mpm,\mathbf{s}^{\prime}\right)\\
			& & \geq & \frac{\epsilon}{4}\epsilon^{O\left(\frac{1}{\epsilon}\right)} \left(1-2^{-\Omega\left(n\epsilon^2\right)}\right).
			\end{align*}
		\end{proof}
		\section{Achievability}
		\label{sec:lower-bound}
		
		\begin{figure}[p]
			\centering
			\includegraphics[scale=0.8]{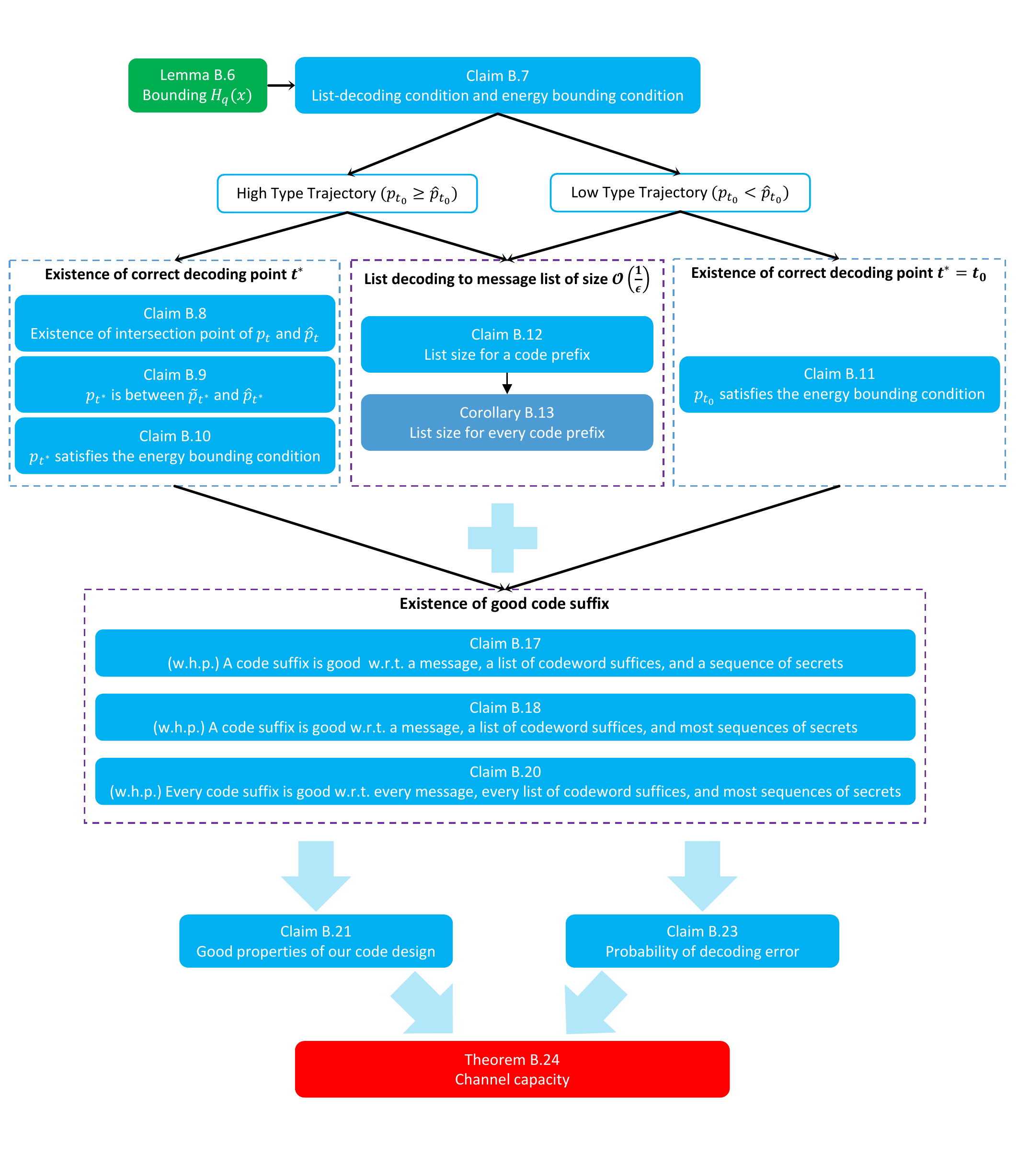}
			\caption{Organization of our claims, corollaries and theorems for the achievability}
			\label{fig:proofs-organization-achievability}
		\end{figure}
		
		We start by summarizing several definitions and claims. The detailed presentations of the definitions and claims are followed by the summary. We depict the flow of our claims, corollaries, and theorems in Figure~\ref{fig:proofs-organization-achievability}.
		
		\begin{enumerate}[label=\textbf{\arabic*}.]
			\item {\bf Preliminary definitions and technical claims}
			\begin{itemize}
				\item Definition~\ref{def:p-t-app}: Defines Calvin's trajectory $\ppp_{\mt}$ with respect to the unerased positions up to $\mt$, which is the number of symbol errors normalized by the number of unerased positions up to $\mt$.
				\item Definition~\ref{def:p-bar-t}: Defines Bob's guess of random noise $\pbar_t$ for deriving the definition of the decoding reference trajectory $\phatt$.
				\item Definition~\ref{def:p-hat-t}: Defines Bob's decoding reference trajectory $\phatt$, which is a revision of the definition given in Section~\ref{sec:model}.
				\item Definition~\ref{def:tr-type} Defines two types of trajectory of Calvin according to $\phat_{\mt_{0}}$.
				\item Definition~\ref{def:p-tilde-t} Defines the energy bounding trajectory $\ptil_{\mt}$, which delimits the smallest value of $\pref{p}$ that meets with the energy bounding condition.
				\item Lemma~\ref{lemma-1}: A technical lemma which gives a certain upper bound on the $q$-ary entropy function.
			\end{itemize}
			\item {\bf The list decoding and energy bounding properties}
			\begin{itemize}
				\item Claim~\ref{claim:list-energy}: This is a central claim which shows that the decoding reference trajectory $\phatt$ satisfies the list-decoding condition and the energy bounding condition.
			\end{itemize}
			\item {\bf Establishing the existence of correct decoding point}
			\begin{itemize}
				\item Claim~\ref{claim-above}: Calvin's trajectory $\pref{p}$ always intersects with the decoding reference trajectory $\phatt$ no later than the second to last chunk.
				\item Claim~\ref{claim-ingap}: For any High Type Trajectory $\pref{p}$, the value of $\pref{p}$ at the chunk end immediately after the intersection of the decoding reference trajectory $\phatt$ with $\pref{p}$ satisfies the energy bounding condition \zt{(Recall that both $\phatt$ and $\ppp_{\mt}$ are defined with respect to unerased positions}). 
				\item Claim~\ref{claim:p-tilde-t}: If $\pref{p}$ is larger than $\tilde{p}_{t}$ at point $t$, then $\pref{p}$ satisfies the energy bounding condition.
				\item Claim~\ref{claim-zero-starting}: At point $t_{0}$, if $p_{t_{0}}$ is approximately $\phat_{t_{0}}$ then it satisfies the energy bounding condition.
			\end{itemize}
			\item {\bf List decoding properties}
			\begin{itemize}
				\item Claim~\ref{thm:b-list-decodable}: A code prefix can be list decoded to a list of messages of size $\blistord$ with high probability.
				\item Corollary~\ref{coro:b-list-decodable}: Every code prefix can be list decoded to a list of messages of size $\blistord$ with high probability.
			\end{itemize}
			\item {\bf Utilizing the energy bounding condition}
			\begin{itemize}
				\item Definition~\ref{def:dist}: Defines the distance between a codeword suffix and a list of codeword suffixes.
				\item Definition~\ref{def:good-1}: Defines certain {\em goodness} properties of a code suffix with respect to a message, a list of codeword suffixes (of messages excluding the transmitted message), and a sequence of secrets.
				\item Definition~\ref{def:good-2}: Defines $\sigma$-goodness property of a code suffix with respect to a message, a list of codeword suffixes (of messages excluding the transmitted message), and most sequences of secrets.
				\item Claim~\ref{claim:good-1}: A code suffix is good with respect to a message, a list of codeword suffixes (of messages excluding the transmitted message), and a sequence of secrets. 			
				\item Claim~\ref{claim:good-2}: A code suffix is $\sigma$-good with respect to a message and a list of codeword suffixes (of messages excluding the transmitted message).
				\item Claim~\ref{claim:good-3}: Every code suffix is $\sigma$-good with respect to every transmitted message and every list of codeword suffixes (of messages excluding the transmitted message).
			\end{itemize}
			\item {\bf Summary and proof of Theorem~\ref{thm:main:capacity}}
			\begin{itemize}
				\item Claim~\ref{claim:prob-good-exist}: With high probability our code $\mcode$ possesses the needed properties.
				\item Claim~\ref{claim:prob-decode}: With high probability Bob succeeds in decoding.
				\item Theorem~\ref{thm:capacity}: Rephrasing of Theorem~\ref{thm:main:capacity} (channel capacity).
			\end{itemize}
		\end{enumerate}
		
		Let $\erate>0$ and $q\geq 2$. Let $\ppp\in\left(\intverr\right)$ be the fraction of symbol errors and $\pstar\in\left(\intvera\right)$ be the fraction of symbol erasures such that $2\ppp+\pstar+\erate\leq\frac{q-1}{q}$.
		
		Let $\echunk=\frac{\erate^2}{9q^{2}}$. Let $\leftcw\in\setchend=\choiceschunk$.

		Assume the received word $\my\in\mathcal{Y}^{n}$ has $np$ symbol errors and $n\pstar$ erasures. For any $t\in\setchend$, let $\lat$ be the number of erasures in $\my$ up to position $t$.
		
		\zt{Let $t_{0}=k_{0}\chunk\in\setchend$ be the smallest integer such that $\mt_0 - \lat[\mt_{0}]\geq n\left(1-\frac{2q}{q-1}\ppp-\frac{q}{q-1}\pstar-\frac{\erate^2}{4}\right)$.
		
		Let $S=\zx{\echunk^3/q^2}$ be the secret rate, namely, $q^{nS}$ is the size of the set $\secr$ of secrets available to Alice.}
		
		\subsection{Preliminaries}
		
		\begin{definition}[Calvin's Trajectory $\ppp_{\mt}$]
			\label{def:p-t-app}
			\zt{Let $\ppp_{\mt} \in[0,1]$ be the actual fraction of symbol errors with respect to the unerased positions in the codeword prefix of $\mx$ with respect to position $t$}.
		\end{definition}
		
		\begin{definition}[Bob's Guess of Random Noise $\pbar_{\mt}$] 
			
			\label{def:p-bar-t}
			\begin{align}
			\label{eq:p-bar-t}
			\pbar_{\mt}=
			\ppp+\frac{\pstar}{2} -\frac{q-1}{2q}\left(1-\frac{\mt-\lat}{n}\right). 
			\end{align} 
		\end{definition}
		
		\begin{definition}[Bob's Decoding Reference Trajectory $\phat_{\mt}$]
			\label{def:p-hat-t}
			Let $\alfq[\mt]=\alfexpq[\mt]$ where $\pbar_{\mt}$ is as in Definition \ref{def:p-bar-t}.
			Then
			\begin{align}
			\label{eq:p-hat-t}
			\phat_{\mt}=
			\begin{cases}
			\frac{\erate^2}{9q^{2}\alpha^2_q(0)}, 
			& (\mt-\lat) \in 
			\left[
			n\left(1-\frac{2q}{q-1}\ppp-\frac{q}{q-1}\pstar-\frac{\erate^2}{4}\right),
			n\left(1-\frac{2q}{q-1}\ppp-\frac{q}{q-1}\pstar\right)
			\right),
			\\
			\frac{\pbar_{\mt}}{\alfq[\mt]}+\epht, 
			& (\mt-\lat) \in 
			\left[
			n\left(1-\frac{2q}{q-1}\ppp-\frac{q}{q-1}\pstar\right),
			n\left(1-\frac{q}{q-1}\pstar\right)
			\right].
			\end{cases}
			\end{align}
		\end{definition}

		\begin{definition}[Trajectory Type]
			\label{def:tr-type}
			For any trajectory $\ppp_{t}$ of Calvin, consider the values of $\ppp_{t}$ and $\phatt$ at position $t=t_{0}$. If $\ppp_{t_{0}} \geq \phat_{t_0}$ then Calvin's trajectory $\ppp_{t}$ is a \textit{High} Type Trajectory, otherwise $\ppp_{t}$ is a \textit{Low} Type Trajectory.
		\end{definition}
		
		\begin{definition}[Energy Bounding Trajectory $\tilde{p}_{t}$]
			\label{def:p-tilde-t}
			Let $\alfq[\mt]=\alfexpq[\mt]$ where $\pbar_{\mt}$ is as in Definition \ref{def:p-bar-t}. Then
			\begin{align}
			\label{eq:p-tilde-t}
			\tilde{p}_{t}=\frac{\pbar_{t}}{\alfq[t]}+\frac{(n-t)\erate^2}{9q^{2}(\mt-\lat)}
			\end{align}
		\end{definition}
		
		\begin{lemma}
			\label{lemma-1}
			Let $q\geq 2$ and $H_q(x)=x\log_q{(q-1)}-x\log_q{x}-(1-x)\log_q{(1-x)}$ for $x\in\left[0,1-1/q\right]$. Then for any $\del\in(0,1/2)$, we have
			\begin{align*}
			H_q(x+\del)<H_q(x)+\frac{2\sqrt{\del}+{\del\ln{(q-1)}}}{\ln{q}}.
			\end{align*}
		\end{lemma}
		
		\begin{proof}
			To prove the lemma, we first show that
			\begin{align*}
			\log (1-x)+2x\geq 0
			\end{align*}
			for $x\in\left[ 0,\frac{1}{2}\right) $ and
			\begin{align*}
			\log (1-x)+2x<0
			\end{align*} 
			for $x\in\left( \frac{1}{2},1\right] $. 
			
			Let $f(x)=\log (1-x)+2x$ where $x\in\left[ 0,1\right] $. Then $f^{\prime}(x)=2-\frac{1}{(1-x)\ln 2}$. Solving $f^{\prime}(x)=0$, we obtain $x=1-\frac{1}{2\ln 2}<\frac{1}{2}$. Then for $x\in\left( 0,1-\frac{1}{2\ln 2}\right) $, $f^{\prime}(x)>0$ and for $x\in\left( 1-\frac{1}{2\ln 2},1\right) $, $f^{\prime}(x)<0$.
			
			Since $f(0)=f\left( \frac{1}{2}\right) =0$, then for $x\in\left[ 0,\frac{1}{2}\right) $ we have $\log (1-x)+2x\geq0$, and therefore,
			\begin{align}
			\label{eq:le-1}
			\log \frac{1}{1-x}\leq 2x.
			\end{align}
			On the other hand, for $x\in\left( \frac{1}{2},1\right] $ we have $\log (1-x)+2x<f\left( \frac{1}{2}\right) =0$, and thus, replacing $(1-x)$ by $x$ we have for $x\in\left[ 0,\frac{1}{2}\right) $
			\begin{align}
			\label{eq:le-11}
			2(1-x)<\log\frac{1}{x}.
			\end{align}
			
			Since $H_q(x)$ is concave, namely, the second derivative of $H_q(x)$ is negative for $x\in\left( 0,1-1/q\right) $, then
			\begin{align*}
			\frac{H_q(x+\del)-H_q(x)}{x+\del-x}<\frac{H_q(\del)-H_q(0)}{\del-0}.
			\end{align*}
			Therefore, we have
			\begin{align}
			H_q(x+\del)-H_q(x) &<H_q(\del)-H_q(0)\nonumber\\
			&= \del\log_q{\frac{1}{\del}} +(1-\del)\log_q{\frac{1}{1-\del}}+\del\log_q{(q-1)}\nonumber\\
			&=\frac{1}{\log{q}} \left(\del\log{\frac{1}{\del}} +(1-\del)\log{\frac{1}{1-\del}}+\del\log{(q-1)}\right)\nonumber\\
			&\leq \frac{1}{\log{q}} \left(\del\log \frac{1}{\del}+(1-\del)2\del+\del\log{(q-1)}\right)\label{eq:le-2}\\
			&< \frac{1}{\log{q}} \left(\del\log \frac{1}{\del}+\del\log \frac{1}{\del}+\del\log{(q-1)}\right)\label{eq:le-3}\\
			&= \frac{1}{\log{q}} \left(2\del\log \frac{1}{\del} +\del\log{(q-1)}\right)\nonumber
			\end{align}
			where \eqref{eq:le-2} follows by \eqref{eq:le-1} and \eqref{eq:le-3} follows by \eqref{eq:le-11}.
			
			Note that $\ln{x}\leq\frac{x-1}{\sqrt{x}}$ for $x\geq 1$ as $g(x)=\frac{x-1}{\sqrt{x}}-\ln{x}$ is monotonically increasing for $x\geq 1$ and $g(1)=0$. Then for $\del\in(0,1/2)$ we have 
			\begin{align}
			\del\ln{\frac{1}{\del}} \leq \del\left(\frac{1}{\sqrt{\del}}-\sqrt{\del}\right) < \sqrt{\del}\label{eq:le-ln}.
			\end{align}
			
			Hence, we have
			\begin{align}
			H_q(x+\del)-H_q(x) & < \frac{1}{\log{q}} \left(2\del\log \frac{1}{\del} +\del\log{(q-1)}\right)\nonumber\\
			& = \frac{1}{\ln{q}} \left(2\del\ln \frac{1}{\del} +\del\ln{(q-1)}\right)\nonumber\\
			& < \frac{1}{\ln{q}} \left(2\sqrt{\del} +\del\ln{(q-1)}\right)\label{eq:le-last}
			\end{align}
			where \eqref{eq:le-last} follows by \eqref{eq:le-ln}.
		\end{proof}
		
		\subsection{The list decoding and energy bounding properties}
		\begin{claim}
			\label{claim:list-energy}
			Let $\alfq=\alfexpq$ where $\pbar\in\left[0,\ppp\right]$. Let $$C=\capopq$$ and $R=C-\erate$. Then for any $t\in\setchend$ and $(\mt-\lat) \in 
			\left[n \left(1-\frac{2q}{q-1}\ppp-\frac{q}{q-1}\pstar-\frac{\erate^2}{4}\right),
			n\left(1-\frac{q}{q-1}\pstar\right)
			\right]$ there exists $\phatt\in\left[0,1-1/q\right]$ such that the following conditions are satisfied.
			\begin{align}
			\label{eq:list-con-c}
			& (\mt-\lat)\left( 1-H_{q}(\phatt)\right)-\frac{n\erate}{4} \geq nR\\
			\label{eq:energy-con-c}
			& n\ppp-(\mt-\lat)\phatt+\errdist\leq\frac{q-1}{2q}\left(n-n\pstar-\mt+\lat\right)
			\end{align}
		\end{claim}
		\begin{proof}
			First note that there exists $\mt\in\setchend$ and $(\mt-\lat) \in 
			\left[
			n\left(1-\frac{2q}{q-1}\ppp-\frac{q}{q-1}\pstar-\frac{\erate^2}{4}\right),n\left(1-\frac{q}{q-1}\pstar\right)
			\right]$ as $\erate^2/4>\echunk$.
			
			Then for $\zt{(\mt-\lat)} \in 
			\left[
			n\left(1-\frac{2q}{q-1}\ppp-\frac{q}{q-1}\pstar\right),n\left(1-\frac{q}{q-1}\pstar\right)
			\right]$, we have $\pbar_{t}\in[0,p]$. Substituting \eqref{eq:p-bar-t} into $\pbar_{t}$ in $n\alfq[t]=n\left(\alfexpq[t]\right)$, we obtain $n\alfq[t]=\mt-\lat$.
			Next, replacing $(\mt-\lat)$ by $n\alfq[t]$ in \eqref{eq:list-con-c} and dividing both sides by $n$, we obtain
			\begin{align}
			\label{eq:con-1-reduced}
			\alfq[t]\left( 1-H_q(\phatt)\right) -\elist\geq R.
			\end{align}
			
			Then, we substitute \eqref{eq:p-hat-t} into $\phatt$ in the \textit{left hand side} (LHS) of \eqref{eq:con-1-reduced} and we get
			\begin{align}
			\lefteqn{\alfq[t] \left( 1-H_q\left( \frac{\pbar_{t}}{\alfq[t]}+\epht\right) \right) -\elist}\nonumber\\
			& &>&\alfq[t] \left( 1-H_q\left( \frac{\pbar_{t}}{\alfq[t]}\right) -\frac{2}{\ln{q}}\sqrt{\epht}-\frac{\ln{(q-1)}}{\ln{q}}\epht\right) -\elist \label{eq:claim-1-left-larger-1}\\
			& &>&\alfq[t] \left( 1-H_q\left( \frac{\pbar_{t}}{\alfq[t]}\right) -\frac{2+\ln{(q-1)}}{\ln{q}}\sqrt{\epht}\right) -\elist\nonumber\\
			& &>&\alfq[t] \left( 1-H_q\left( \frac{\pbar_{t}}{\alfq[t]}\right) - \frac{\erate}{q\alfq[t]}\right) -\elist\label{eq:claim-1-left-larger-2}\\
			& &>& \alfq[t]\left( 1-H_q\left( \frac{\pbar_{t}}{\alfq[t]}\right) \right) -\erate \nonumber\\
			& &\geq & \capopq-\erate\nonumber\\
			& &=& C-\erate\nonumber\\
			& &=& R\nonumber
			\end{align}
			where \eqref{eq:claim-1-left-larger-1} follows from Lemma~\ref{lemma-1} and \eqref{eq:claim-1-left-larger-2} follows by $\frac{2+\ln{(q-1)}}{\ln{q}}<3$ for $q\geq 2$. 
			
			
			For $(\mt-\lat)\in
			\left[
			n\left(1-\frac{2q}{q-1}\ppp-\frac{q}{q-1}\pstar-\frac{\erate^2}{4}\right),
			n\left(1-\frac{2q}{q-1}\ppp-\frac{q}{q-1}\pstar\right)
			\right)$, we have $$\frac{\mt-\lat}{n} \geq 1-\frac{2q}{q-1}\ppp-\frac{q}{q-1}\pstar-\frac{\erate^2}{4}=\alpha_q(0)-\frac{\erate^2}{4}.$$ 
			
			Then
			\begin{align}
			\frac{\mt-\lat}{n}\left( 1-H_{q}(\phatt)\right)-\elist & \geq \left(\alpha_q(0)-\frac{\erate^2}{4}\right)\left( 1-H_{q}(\phatt)\right)-\elist\nonumber\\
			& = \left(\alpha_q(0)-\frac{\erate^2}{4}\right)\left( 1-H_{q}\left(\frac{\erate^2}{9q^{2}\alpha^2_q(0)}\right)\right)-\elist\nonumber\\
			& > \left(\alpha_q(0)-\frac{\erate^2}{4}\right)\left(1-\frac{\erate}{q\alpha_q(0)}\right)-\elist \label{eq:claim-1-small-t}\\
			& = \left(\alpha_q(0)-\frac{\erate^2}{4}\right)-\left(\alpha_q(0)-\frac{\erate^2}{4}\right)\frac{\erate}{q\alpha_q(0)}-\elist\nonumber\\
			& > \alpha_q(0)-\frac{\erate^2}{4} - \frac{3\erate}{4}\nonumber\\
			& > \capopq-\erate\nonumber\\
			& = R\nonumber
			\end{align}
			where \eqref{eq:claim-1-small-t} follows from Lemma~\ref{lemma-1}, $\frac{2+\ln{(q-1)}}{\ln{q}}<3$ for $q\geq 2$. 
			
			Thus far we have satisfied condition~\eqref{eq:list-con-c} in our claim. To see condition~\eqref{eq:energy-con-c}, we substitute \eqref{eq:p-hat-t} into $\phatt$ in the LHS of \eqref{eq:energy-con-c}, and note that for $(\mt-\lat) \in
			\left[ n\left(1-\frac{2q}{q-1}\ppp-\frac{q}{q-1}\pstar\right),
			n\left(1-\frac{q}{q-1}\pstar\right)
			\right]$, we have $\alfq[t]=(\mt-\lat)/n$, and therefore,
			\begin{align}
			n\ppp-(\mt-\lat)\left( \frac{\pbar_{t}}{\alfq[t]}+\epht\right) +\errdist 
			& = n\ppp-n\pbar_{t}-\frac{n^2\erate^2}{9q^{2}(\mt-\lat)}+\errdist
			\nonumber\\
			& < n\ppp-n\pbar_{t}
			\nonumber\\
			& = \frac{q-1}{2q}(n-\mt+\lat)-\frac{n\pstar}{2}
			\label{eq:con-2-less-1}\\
			& < \frac{q-1}{2q}(n-n\pstar-\mt+\lat)
			\label{eq:claim-1-energy-large}
			\end{align}
			where \eqref{eq:con-2-less-1} follows by substituting \eqref{eq:p-bar-t} into $\pbar_{t}$.
			
			For $(\mt-\lat) \in
			\left[
			n\left(1-\frac{2q}{q-1}\ppp-\frac{q}{q-1}\pstar-\frac{\erate^2}{4}\right),
			n\left(1-\frac{2q}{q-1}\ppp-\frac{q}{q-1}\pstar\right)
			\right)$, we have $\phatt=\frac{\erate^2}{9q^{2}\alpha^2_q(0)}$. 
			
			Let $f(\mt-\lat)=\frac{\pbar_{t}}{\alfq[t]}+\epht$ for $\mt-\lat \geq n\left(1-\frac{2q}{q-1}\ppp-\frac{q}{q-1}\pstar-\frac{\erate^2}{4}\right)$. As $f(\mt-\lat)$ is a monotonically increasing in $(\mt-\lat)$ for $(\mt-\lat) \in
			\left[
			n\left(1-\frac{2q}{q-1}\ppp-\frac{q}{q-1}\pstar-\frac{\erate^2}{4}\right),
			n\left(1-\frac{2q}{q-1}\ppp-\frac{q}{q-1}\pstar\right)
			\right)$, we have 
			$\frac{\pbar_{t}}{\alfq[t]}+\epht < \frac{\erate^2}{9q^{2}\alpha^2_q(0)}$. 
			Therefore,
			\begin{align}
			n\ppp-(\mt-\lat)\cdot\frac{\erate^2}{9q^{2}\alpha^2_q(0)} +\errdist & < n\ppp-(\mt-\lat)\left( \frac{\pbar_{t}}{\alfq[t]}+\epht\right) +\errdist\nonumber\\
			& < \frac{q-1}{2q}(n-n\pstar-\mt-\lat)
			\label{eq:claim-1-energy-small}
			\end{align}
			where \eqref{eq:claim-1-energy-small} follows by \eqref{eq:claim-1-energy-large}.
		\end{proof}
		
		%
		\subsection{Establishing the existence of correct decoding point}
		First we show that $\phatt$ must eventually be greater than $\ppp_{t}$.
		
		\begin{claim}
			If $\mt-\lat = n-\frac{q}{q-1}n\pstar-\chunk$, then $(\mt-\lat)\phatt\geq n\ppp$.
			\label{claim-above}
		\end{claim}
		
		\begin{proof}
			Since $(\mt-\lat) \in \left[
			n\left(1-\frac{2q}{q-1}\ppp-\frac{q}{q-1}\pstar\right),
			n\left(1-\frac{q}{q-1}\pstar\right)
			\right]$
			then $$\phatt=\frac{\pbar_{\mt}}{\alfq[\mt]}+\epht[\mt].$$
			Hence,
			\begin{align}
			(\mt-\lat)\phatt 
			& = n\pbar_{\mt}+\frac{n^2\erate^2}{9q^{2}\left(\mt-\lat\right)}
			\label{eq:claim-above} \\
			& >
			n\pbar_{\mt}+\frac{n\erate^2}{9q^{2}}
			\nonumber\\
			& = np-\frac{(q-1)\chunk}{2q} + \frac{n\erate^2}{9q^{2}}
			\label{eq:claim-above-2} \\
			& > np -\frac{\chunk}{2} + \frac{n\erate^2}{9q^{2}}
			\nonumber \\
			& > np
			\nonumber
			\end{align}
			where \eqref{eq:claim-above} follows by $\alfq[\mt]=\left(\mt-\lat\right)/n$ and \eqref{eq:claim-above-2} follows by substituting the expression of $\pbar_{\mt}$.
		\end{proof}
		
		\begin{claim}
			\label{claim-ingap}
			For any $t\in\setchend$ and $(\mt-\lat) \in
			\left[ n\left(1-\frac{2q}{q-1}\ppp-\frac{q}{q-1}\pstar-\frac{\erate^2}{4}\right)+\chunk,
			n\left(1-\frac{q}{q-1}\pstar\right)
			\right]$, if $\ppp_{\mt-\chunk}>\phat_{\mt-\chunk}$, 
			then $\ppp_{t}>\tilde{p}_{t}$.
		\end{claim}
		
		\begin{proof}
			For $(\mt-\lat) \in
			\left[ n\left(1-\frac{2q}{q-1}\ppp-\frac{q}{q-1}\pstar\right)+\chunk,
			n\left(1-\frac{q}{q-1}\pstar\right)
			\right]$, we have
			\begin{align}
			\phatt - \ppp_{t} & \leq \phatt - \frac{(\mt-\chunk-\lat[\mt-\chunk])\ppp_{t-\chunk}}{\mt-\lat}
			\nonumber \\
			& < \phatt - \frac{(\mt-\chunk-\lat[\mt-\chunk])\phat_{t-\chunk}}{\mt-\lat}
			\label{eq:ingap-g1} \\
			& = \left(\frac{\pbar_{t}}{\alfq[t]}+\epht\right)
			-
			\frac{\mt-\chunk-\lat[\mt-\chunk]}{\mt-\lat} \left(\frac{\pbar_{t-\chunk}}{\alfq[t-\chunk]} + \frac{\erate^2}{9q^{2}\alpha^2_q(\pbar_{t-\chunk})}\right)
			\label{eq:ingap-phatt} \\
			& = \frac{n}{\mt-\lat}\left(\pbar_{t}-\pbar_{t-\chunk}\right)
			+
			\frac{n^2\erate^2}{9q^{2}}\left(\frac{1}{(\mt-\lat)^2}-\frac{1}{(\mt-\chunk-\lat[\mt-\chunk])(\mt-\lat)}\right)
			\label{eq:ingap-alpha} \\
			& < \frac{n}{\mt-\lat}\left(\pbar_{t}-\pbar_{t-\chunk}\right)
			\nonumber \\
			& = \frac{n}{\mt-\lat}\cdot\frac{q-1}{2q}\echunk
			\label{eq:ingap-pbar} \\
			& < \frac{\chunk}{2(\mt-\lat)}
			\nonumber
			\end{align}
			where \eqref{eq:ingap-g1} follows by using the fact that $\ppp_{n-\chunk}>\phat_{n-\chunk}$, \eqref{eq:ingap-phatt} following by substituting the expression of $\phatt$, \eqref{eq:ingap-alpha} follows by $\alfq[t]=(\mt-\lat)/n$, and \eqref{eq:ingap-pbar} follows by substituting the expression of $\pbar_{t}$.
			
			On the other hand, since $\ptil_{t}=\frac{\pbar_{t}}{\alfq[t]}
			+
			\frac{(n-t)\erate^2}{9q^{2}(\mt-\lat)}
			=
			\phatt
			-
			\frac{n^2\erate^2-(n-t)(\mt-\lat)\erate^2}{9q^{2}(\mt-\lat)^2}$, then
			\begin{align}
			\phatt  - \ptil_{t} & =  \frac{n^2\erate^2-(n-t)(\mt-\lat)\erate^2}{9q^{2}(\mt-\lat)^2}
			\nonumber \\
			& = \frac{n^2\erate^2}{9q^{2}(\mt-\lat)^2}
			-
			\frac{(2n-t)\erate^2}{9q^{2}(\mt-\lat)}
			+
			\frac{n\erate^2}{9q^{2}(\mt-\lat)}
			\nonumber \\
			& > \frac{n\erate^2}{9q^{2}(\mt-\lat)}
			\label{eq:ingap-n-t} \\
			& \geq \frac{\chunk}{\mt-\lat}
			\nonumber \\
			& > \phatt - \ppp_{t}
			\nonumber
			\end{align}
			where \eqref{eq:ingap-n-t} follows by $n^2>t(2n-t)$.
			Since $\phatt - \ptil_{t} > \phatt - \ppp_{t}$, it follows that $\ppp_{t} > \ptil_{t}$.
			
			To show $\ppp_{t}>\ptil_{t}$ for $(\mt-\lat) \in
			\left[
			n\left(1-\frac{2q}{q-1}\ppp-\frac{q}{q-1}\pstar-\frac{\erate^2}{4}\right)+\chunk,
			n\left(1-\frac{2q}{q-1}\ppp-\frac{q}{q-1}\pstar\right)+\chunk
			\right)
			$, we
			let $f(\mt-\lat)=\frac{\pbar_{t}}{\alfq[t]}+\epht$ for $t\geq n\left(1-\frac{2q}{q-1}\ppp-\frac{q}{q-1}\pstar-\frac{\erate^2}{4}\right)$. As $f(\mt-\lat)$ is monotonically increasing for $(\mt-\lat) \in
			\left[
			n\left(1-\frac{2q}{q-1}\ppp-\frac{q}{q-1}\pstar-\frac{\erate^2}{4}\right),
			n\left(1-\frac{2q}{q-1}\ppp-\frac{q}{q-1}\pstar\right)+\chunk
			\right)
			$, we have $\phatt \geq f(\mt-\lat)$.
			Therefore, 
			\begin{align*}
			\phatt-\ptil_{t} & \geq f(\mt-\lat) - \ptil_{t}
			\\ 
			& = \frac{n^2\erate^2-(n-t)(\mt-\lat)\erate^2}{9q^{2}(\mt-\lat)^2}
			\\
			& > \frac{\chunk}{\mt-\lat}
			\end{align*} for $(\mt-\lat) \in
			\left[
			n\left(1-\frac{2q}{q-1}\ppp-\frac{q}{q-1}\pstar-\frac{\erate^2}{4}\right)+\chunk,
			n\left(1-\frac{2q}{q-1}\ppp-\frac{q}{q-1}\pstar\right)+\chunk
			\right)
			$.
			
			Next, we consider the difference between $\phatt$ and $\ppp_{t}$.
			
			If $(\mt-\lat) \in
			\left[
			n\left(1-\frac{2q}{q-1}\ppp-\frac{q}{q-1}\pstar-\frac{\erate^2}{4}\right)+\chunk,
			n\left(1-\frac{2q}{q-1}\ppp-\frac{q}{q-1}\pstar\right)
			\right)
			$, then $\phatt = \frac{\erate^2}{9q^{2}\alpha^2_q(0)}$, and thus,
			\begin{align*}
			\phatt - \ppp_{t} & < \phatt - \frac{(\mt-\chunk-\lat[\mt-\chunk])\phat_{t-\chunk}}{\mt-\lat}
			\\
			& = \phatt - \frac{(\mt-\chunk-\lat[\mt-\chunk])\phatt}{\mt-\lat}
			\\
			& < \frac{\chunk\phatt}{\mt-\lat}.
			\end{align*}
			If $(\mt-\lat) \in
			\left[
			n\left(1-\frac{2q}{q-1}\ppp-\frac{q}{q-1}\pstar\right),
			n\left(1-\frac{2q}{q-1}\ppp-\frac{q}{q-1}\pstar\right)+\chunk
			\right)
			$, then
			\begin{align*}
			\phatt - \ppp_{t} & < \phatt - \frac{(\mt-\chunk-\lat[\mt-\chunk])\phat_{t-\chunk}}{\mt-\lat}
			\\
			& \leq \phatt - \frac{(\mt-\chunk-\lat[\mt-\chunk])f(\mt-\lat)}{\mt-\lat}
			\\
			& = \left(\frac{\pbar_{t}}{\alfq[t]}+\epht\right)
			-
			\frac{\mt-\chunk-\lat[\mt-\chunk]}{\mt-\lat} \left(\frac{\pbar_{t-\chunk}}{\alfq[t-\chunk]} + \frac{\erate^2}{9q^{2}\alpha^2_q(\pbar_{t-\chunk})}\right)
			\\
			& < \frac{\chunk}{2(\mt-\lat)}.
			\end{align*}
			
			Hence, for any $(\mt-\lat)\in
			\left[
			n\left(1-\frac{2q}{q-1}\ppp-\frac{q}{q-1}\pstar-\frac{\erate^2}{4}\right)+\chunk,
			n\left(1-\frac{2q}{q-1}\ppp-\frac{q}{q-1}\pstar\right)+\chunk
			\right)
			$, we have $\phatt - \ppp_{t} < \frac{\chunk}{\mt-\lat} < \phatt - \ptil_{t}$, and it follows that $\ppp_{t} > \ptil_{t}$.
		\end{proof}
		
		\begin{claim}
			\label{claim:p-tilde-t}
			Let $\suff{p}$ be the portion of symbol errors in the codeword $\mathbf{x}$ with respect to the unerased positions between position $\leftcw+1$ and $n$ for $\zx{\mt-\lat} \in 
			\left[
			n\left(1-\frac{2q}{q-1}\ppp-\frac{q}{q-1}\pstar-\frac{\erate^2}{4}\right),n\left(1-\frac{q}{q-1}\pstar\right)
			\right]$. If $p_t>\tilde{p}_t$, then $\suff{p}<\zx{\frac{q-1}{2q}-\frac{\erate^2}{9q^{2}}-\frac{n\pstar}{2q(n-\mt-n\pstar+\lat)}}$.
		\end{claim}
		
		\begin{proof}
			By the definition of $\ppp_{h}$, we have $\ppp_{h}=\frac{n\ppp-(\mt-\lat)\ppp_{t}}{n-n\pstar-\mt+\lat}$. Since $p_t>\tilde{p}_t$, then
			\begin{align}
			\ppp_{h} & < \frac{n\ppp-(\mt-\lat)\ptil_{t}}{n-n\pstar-\mt+\lat}
			\nonumber \\
			& = \frac{1}{n-n\pstar-\mt+\lat} \left(np-n\pbar_{t}-\frac{(n-t)\erate^2}{9q^{2}}\right)
			\label{eq:claim-bound-ptil} \\
			& \zx{= \frac{1}{n-n\pstar-\mt+\lat}\left(\frac{q-1}{2q}(n-\mt+\lat)-\frac{n\pstar}{2}-\frac{(n-t)\erate^2}{9q^{2}}\right)}
			\label{eq:claim-bound-pbar} \\
			& \zx{< \frac{q-1}{2q}-\frac{\erate^2}{9q^{2}}-\frac{n\pstar}{2q(n-\mt-n\pstar+\lat)}}
			\nonumber
			\end{align}
			where \eqref{eq:claim-bound-ptil} follows by \eqref{eq:p-tilde-t} and $\alfq[t]=(\mt-\lat)/n$ and \eqref{eq:claim-bound-pbar} follows by \eqref{eq:p-bar-t}.
		\end{proof}

		\begin{claim}
			\label{claim-zero-starting}
			Let $k_0=\left\lceil
			\frac{1-2pq/(q-1)-\pstar q/(q-1)-\erate^2/4}{\echunk}+\frac{\lat[\mt_{0}]}{\chunk}
			\right\rceil$ and $t_0=k_0\chunk$. Then for any
			$p_{t_0}\in\left[ 0,\zx{\phat_{\mt_{0}}}\right]  $ \zx{where $\phat_{\mt_{0}}=\frac{\erate^2}{9q^{2}\alpha^2_q(0)}$}, 
			we have 
			\begin{align}
			n\ppp-(\mt_{0}-\lat[\mt_{0}]) \ppp_{t_0}+\frac{(n-t_0)\erate^2}{9q^{2}} \leq \frac{q-1}{2q}(n-n\pstar-\mt_{0}+\lat[\mt_{0}])
			\nonumber
			\end{align}
		\end{claim}
		
		\begin{proof}
			Since $t_{0}=k_{0}\chunk < n\left(1-2pq/(q-1)-\pstar q/(q-1)-\erate^2/4+\echunk\right)+\lat[\mt_{0}]$, then
			\begin{align}
			\frac{q-1}{2q}(n-n\pstar-\mt_{0}+\lat[\mt_{0}]) & > \frac{q-1}{2q}\left(n
			\left(\frac{2pq}{q-1}+\frac{\pstar q}{q-1}+\frac{\erate^2}{4}-\echunk-\pstar
			\right)
			\right)
			\nonumber \\
			& > \frac{q-1}{2q}\left(n
			\left(\frac{2pq}{q-1}+\frac{\erate^2}{4}-\echunk
			\right)
			\right)
			\nonumber \\
			& > np+\frac{n\erate^2}{9q^{2}}
			\nonumber \\
			& > n\ppp-\zx{(\mt_{0}-\lat[\mt_{0}])} \ppp_{t_0}+\frac{(n-t_0)\erate^2}{9q^{2}}.
			\nonumber
			\end{align}
		\end{proof}
		
		\subsection{List decoding properties}
		\begin{claim}
			\label{thm:b-list-decodable}
			Let $\dl>0$ and $S=\zx{\echunk^3/q^2}$.
			Let $(\mt-\lat) \in 
			\left[n \left(1-\frac{2q}{q-1}\ppp-\frac{q}{q-1}\pstar-\frac{\erate^2}{4}\right),
			n\left(1-\frac{q}{q-1}\pstar\right)
			\right]$ and $t=k\chunk\in\setchend$.
			If $(\mt-\lat)\left( 1-H_{q}(\phatt)\right)-\frac{n\erate}{4} \geq nR$, then with probability at least $1-q^{-\dl}$ over code design, the code $\megacode$ is list-decodable for $(\mt-\lat)\phatt$ symbol errors with list size 
			\begin{align*}
			L = \frac{\mt-\lat+\dl}{(\mt-\lat)\left(1-H_q\left(\phatt\right)\right) - nR - n\echunk^2\zx{/q^2}}.
			\end{align*}
		\end{claim}
		
		\begin{proof}
			The proof follows ideas in \cite[Thm. 10.3]{guruswami2001list}, and is modified slightly to correspond to stochastic codes. We stress that although the code is stochastic and each message corresponds to several codewords, we analyze the number $L$ of different {\em messages} with codewords that fall into a Hamming ball of limited radius.
			The number of potential codewords in $k$ chunks is $\left( q^{\chunk}\right) ^k=q^{k\chunk}=\zt{q^{\mt}}$.
			As $\phatt \leq 1-1/q$, the number of words of length $(\mt-\lat)$ in a Hamming ball of radius $(\mt-\lat)\phatt$ is at most
			\begin{align*}
			\sum^{(\mt-\lat)\phatt}_{i=0}\binom{\mt-\lat}{i}(q-1)^{i} & <q^{(\mt-\lat)H_q\left( \phatt\right) }.
			\end{align*}
			
			We study the number of different messages corresponding to codewords that may lie in such a ball.
			Each message $m$ corresponds to \zt{at most $q^{nS/\echunk}$} codewords. 
			Since the encoding of each message is independent of other messages, the probability that there exist more than $L$ messages with corresponding codewords of length $(\mt-\lat)$ all of which lie in the Hamming ball of radius $(\mt-\lat)\phatt$ centered at a received word of length $(\mt-\lat)$ is at most
			
			\begin{align}
			\binom{q^{nR}}{L+1}
			\cdot
			\zt{\left(q^{nS/\echunk}\right)}^{L+1}
			\cdot 
			\left( \frac{q^{(\mt-\lat)H_q\left( \phatt\right) }}{q^{(\mt-\lat)}}\right) ^{(L+1)} 
			& < q^{\left( nR+n\echunk^2\zx{/q^2}\right) (L+1)}\left( \frac{q^{(\mt-\lat)H_q\left( \phatt\right) }}{q^{(\mt-\lat)}}\right) ^{(L+1)}
			\nonumber \\
			& = q^{\left[\left( nR+n\echunk^2\zx{/q^2}\right) - (\mt-\lat)\left( 1-H_q\left( \phatt\right) \right)\right](L+1)}.
			\nonumber
			\end{align}
			
			Thus, the probability that the received word of $k$ chunks is list-decoded to a list of size greater than $L$ is at most
			\begin{align}
			q^{(\mt-\lat)} \cdot q^{\left[\left( nR+n\echunk^2\zx{/q^2} \right) - (\mt-\lat)\left( 1-H_q\left( \phatt\right) \right)\right](L+1)}.
			\label{eq:thm-1-prob-large-list}
			\end{align}
			
			To quantify \eqref{eq:thm-1-prob-large-list}, we study
			\begin{align}
			(\mt-\lat) + \left[\left( nR+n\echunk^2\zx{/q^2} \right) - (\mt-\lat)\left( 1-H_q\left( \phatt\right) \right)\right](L+1) < -\dl
			\label{eq:thm-1-prob-large-list-pow}
			\end{align}
			
			Since $(\mt-\lat)\left( 1-H_{q}(\phatt)\right)-\frac{n\erate}{4} \geq nR$, then
			\begin{align}
			(\mt-\lat)\left( 1-H_{q}(\phatt)\right) & \geq nR + \frac{n\erate}{4}
			\nonumber \\
			& > nR+n\echunk^2\zx{/q^2}.
			\nonumber
			\end{align}
			
			Hence, solving \eqref{eq:thm-1-prob-large-list-pow} for $L$ we have
			\begin{align}
			L > \frac{\mt-\lat+\dl}{(\mt-\lat)\left(1-H_q\left(\phatt\right)\right) - nR - n\echunk^2\zx{/q^2}} - 1.
			\label{eq:thm-1-list-size}
			\end{align}
			
			Therefore, if $L$ satisfies \eqref{eq:thm-1-list-size} the code $\megacode$ is $L$-list decodable with probability at least $1-q^{-\dl}$.
		\end{proof}
		
		\begin{corollary}
			\label{coro:b-list-decodable}
			Let $\dl=3\log_q n$. 
			Let $(\mt-\lat) \in 
			\left[n \left(1-\frac{2q}{q-1}\ppp-\frac{q}{q-1}\pstar-\frac{\erate^2}{4}\right),
			n\left(1-\frac{q}{q-1}\pstar\right)
			\right]$ and $t=k\chunk\in\setchend$.
			Then with probability at least $1-\frac{1}{n}$ over code design, for any $\mt$ such that $(\mt-\lat)\left( 1-H_q(\phatt)\right)-\frac{n\erate}{4} \geq nR$, the code $\megacode$ is $L$-list decodable for $(\mt-\lat)\phatt$ symbol errors with list size
			\begin{align*}
			L = \frac{\mt-\lat+3\log_q n}{(\mt-\lat)\left(1-H_q\left( \phatt\right)\right) -nR-n\echunk^2\zx{/q^2} }=\blistord.
			\end{align*}
		\end{corollary}
		
		\begin{proof}
			By Claim~\ref{thm:b-list-decodable}, with probability $1-q^{-{ 3\log_q n}}$ the code $\megacode$ is $L$-list decodable with list size $L$ being
			\begin{align*}
			\frac{\mt-\lat+{3\log_q n}}{(\mt-\lat)\left(1-H_q\left( \phatt\right)\right) - nR-n\echunk^2\zx{/q^2}}
			\end{align*}
			
			Therefore, the probability that the code is decoded to a list of size greater than $L$ is at most $q^{-{3\log_q n}}=\frac{1}{{n^3}}$.
			
			Since $k<n$ and $(\mt-\lat)\phatt<\mt-\lat<n$, the probability that the code $\megacode$ is $L$-list decodable for any $k$ chunks is at least
			\begin{align*}
			1-n\cdot n\cdot\frac{1}{{n^3}}=1-\frac{1}{{n}}
			\end{align*}
			
			In addition, since $(\mt-\lat)\left( 1-H_q(\phatt)\right)-\frac{n\erate}{4} \geq nR$, we have $(\mt-\lat)\left( 1-H_q(\phatt)\right) - nR-n\echunk^2\zx{/q^2} > n\left(\erate/4-\echunk^2\zx{/q^2}\right)$. 
			Thus, we obtain
			\begin{align*}
			L<\frac{1+O\left( \frac{\log_q n}{n}\right)}{\erate/4-\echunk^2\zx{/q^2}}=\blistord
			\end{align*}
		\end{proof}
		
		\subsection{Utilizing the energy bounding condition}
		Unless otherwise specified, for any $t\in\setchend=\choiceschunk$, integer $k=\frac{\mt}{\chunk}$ is the number of chunks in the prefix of a code (or codeword) with respect to position $t$ and integer $l=\chunknum-\frac{\mt}{\chunk}=\chunknum-k$ is the number of chunks in the suffix of a code (or codeword) with respect to position $t$.
		
		\begin{definition}
			\label{def:dist}
			A codeword suffix, $\rmegacodw$, is of \textit{distance} $d$ from a set of codeword suffixes if the Hamming distance between the suffix $\rmegacodw$ and any suffix in the given set is at least $d$.
		\end{definition}
		
		In what follows we will define properties of our code with respect to a list of {\em codeword suffixes} $\clist$. This list consists of all the {codeword suffixes} corresponding to the $\mlists$ messages in $\mlist$ obtained by Bob in the list decoding phase of his decoding, {\it excluding} the true message $m$ Alice wishes to communicate to Bob, if it is indeed in the list $\mlist$ (it may not be, if $p_t > \phatt$ for the $t$ under consideration). Hence the size $\clists$ of $\clist$ is at most  $q^{nSl} \cdot \mlists$ (if the true message $m \notin \mlist$), and is at most $q^{nSl} \cdot (\mlists-1)$ (if the true message $m \in \mlist$).
		
		\begin{definition}
			\label{def:good-1}
			A code suffix, $\rmegacode$, is \textit{good} with respect to a list $\clist$ of codeword suffixes, a message $m$, and a sequence of $l$ secrets $\left( s_{k+1},s_{k+2},\cdots,s_{\chunknum}\right) $, if the codeword suffix, $\rmegacodw$, is of distance more than $\frac{(n-t)(q-1)}{q}-\zx{\frac{(n-t)2\erate^2}{9q^{3}}}$ from the list $\clist$.
		\end{definition}
		
		\begin{definition}
			\label{def:good-2}
			A code suffix, $\rmegacode$, is \textit{$\sigma$-good} with respect to a list $\clist$ of codeword suffixes and a message $m$, if the code suffix, $\rmegacode$, is good with respect to the message $m$, the list $\clist$, and a $(1-\sigma)$ portion of sequences of $l$ secrets in the set $\secr^{l}$.
		\end{definition}
		
%
		
		\begin{claim}
			\label{claim:good-1} 
			Let $(s_{k+1},s_{k+2},\cdots,s_{\chunknum})\in\secr^{l}$ be a sequence of $l=\chunknum-k$ secrets.
			With probability greater than $1-q^{-\delta(n-t)}$ over code design, a code suffix, $\rmegacode$, is good with respect to message $m$, the list $\clist$, and the secrets $(s_{k+1},s_{k+2},\cdots,s_{\chunknum})$, where $\delta=\zx{\echunk^2/q^2}$ and $S=\zx{\echunk^3/q^2}$.
		\end{claim}
		
		\begin{proof}
			Let $\cwlist$ be the list $\clist$ of codeword suffixes.
			Note that $\clists=q^{nSl} \cdot \blistord$.
			Define the forbidden region with respect to the list $\clist$ as
			\begin{align*}
			F_{\clist}=\bigcup_{i=1}^{{\clists}} B\left( \mathbf{x}_i,r \right) 
			\end{align*}
			where $B\left( \mathbf{x}_i,r \right)$ is the Hamming ball with center $\mathbf{x}_i$ and radius $r=\frac{(n-t)(q-1)}{q}-\zx{\frac{(n-t)2\erate^2}{9q^{3}}}$. 
			We depict the notion of the forbidden region in Figure~\ref{fig:forbid-1}.
			\begin{figure}[htbp]
				\centering
				\includegraphics[scale=0.7]{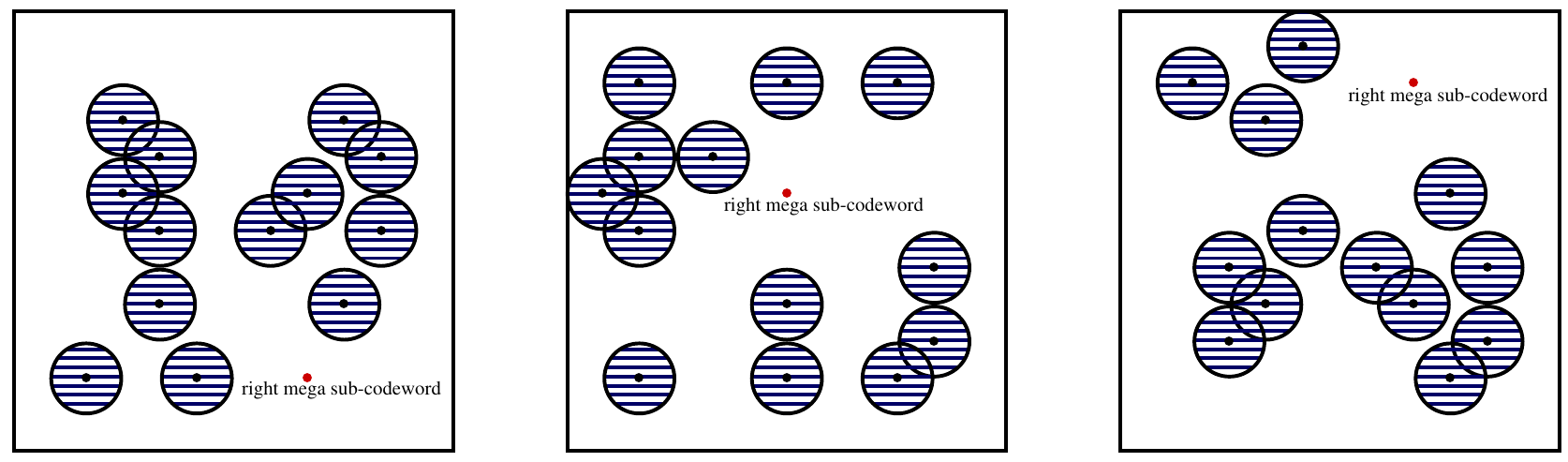}
				\caption[Three realizations of forbidden regions]{Three realizations of forbidden regions: In each realization, shaded disks correspond to the forbidden region and the isolated red point is a codeword suffix outside the forbidden region.}
				\label{fig:forbid-1}
			\end{figure}
			
			Since the size of the list $\clist$ is $\clists$, the number of words of length $(n-t)$ in the forbidden region $F_{\clist}$ can be determined as
			\begin{align}
			\clists\sum_{i=0}^{r}\binom{n-t}{i}(q-1)^i & < \clists q^{(n-t) H_q\left( \frac{q-1}{q}-\zx{\frac{2\erate^2}{9q^{3}}}\right) }
			\nonumber \\
			& < \clists q^{(n-t) \left( 1-\zx{\frac{2\echunk^2}{(q-1)\ln{q}}} \right) }
			\label{eq:size-forbidden-region2} \\
			&\zt{ = q^{(n-t)\left( \frac{\log_q \clists}{n-t} + \left( 1-\zx{\frac{2\echunk^2}{(q-1)\ln{q}}} \right) \right)}}
			\label{eq:size-forbidden-region3}
			\end{align}
			where \eqref{eq:size-forbidden-region2} follows from the Taylor series of the $q$-ary entropy function in a neighborhood of $1-1/q$, i.e., $H_q(x)=1-\frac{q-1}{2q\ln{q}}\sum\limits_{i=1}^{\infty}\frac{(q-1)^{2i-1}+1}{(2i-1)i}\left(1-\frac{q}{q-1}x\right)^{2i}$, and substitution of $\echunk=\frac{\erate^2}{9q^{2}}$.
			
			For sufficiently large $n$ and $S=\zx{\echunk^3/q^2}$, we have for some constant $c$ that
			$$
			\zt{\frac{2\echunk^2}{(q-1)\ln{q}}}-\frac{\log_q \clists}{n-t} = \zt{\frac{2\echunk^2}{(q-1)\ln{q}}}-\frac{S}{\echunk} -\frac{\log_q (c/\epsilon)}{n-t} > \zx{\echunk^2/q^2} = \delta.
			$$
			It follows that
			\begin{align}
			\frac{\log_q{\clists}}{n-t}+\left(1-\zt{\frac{2\echunk^2}{(q-1)\ln{q}}}\right) < 1-\delta 
			\label{eq:size-forbidden-dl}
			\end{align}
			
			Substituting \eqref{eq:size-forbidden-dl} into \eqref{eq:size-forbidden-region3}, we have
			\begin{align}
			\clists \sum_{i=0}^{r}\binom{n-t}{i}(q-1)^i < q^{(n-t)(1-\delta)}
			\label{eq:size-forbidden-region}
			\end{align}
			
			Let $\rmegacodw$ be a codeword suffix corresponding to message $m$. If the codeword suffix is not in the region $F_{\clist}$, then by Definition~\ref{def:good-1}, the code suffix $\rmegacode$ is good with respect to the message $m$, the list $\clist$, and the secrets $(s_{k+1},s_{k+2},\cdots,s_{\chunknum})$. Therefore, the probability over $\rmegacode$ that codeword suffix $\rmegacodw$ does not lie in the forbidden region $F_{\clist}$ is 
			\begin{align}
			\mathbb{P}\left[ \rmegacodw\notin F_{\clist}\right] & > \frac{q^{n-t}-q^{(n-t)(1-\delta)}}{q^{n-t}}
			\nonumber \\
			& = 1-q^{-(n-t)\delta}
			\nonumber
			\end{align}
		\end{proof}
		
		\begin{claim}
			\label{claim:good-2}
			With probability larger than $1-q^{-n^2}$ over code design, a code suffix $\rmegacode$ of length $l=\chunknum-k$ is $\sigma$-good with respect to message $m$ and the list $\clist$, where $\sigma=\egood$.
		\end{claim}
		
		\begin{proof}
			Let $\secr=\left[ q^{nS} \right]$ be the set of integers between 0 and $q^{nS}-1$. We start by considering a partition of the set of codeword suffixes corresponding to message $m$ into $\secr^{l-1}$ disjoint subsets.  Specifically, we partition the set of secrets $\secr^l$ into $\secr^{l-1}$ disjoint sets. 
			Each set is indexed by an element $(s_{k+2},\dots,s_{\chunknum})$ in $\secr^{l-1}$.
			The set $\secr_\mathbf{s^*}$ corresponding to $\mathbf{s^*}=(s^*_{k+2},\dots,s^*_{\chunknum})$ equals:
			$$
			\secr_\mathbf{s^*}=
			\left\lbrace
			\mathbf{s}=(a,s^*_{k+2}+a,\dots,s^*_{\chunknum}+a) \mid a \in \left[ q^{nS} \right]
			\right\rbrace
			$$
			where addition is done modulo $q^{nS}$.
			It holds that
			$$
			\secr^l=\bigcup_{\mathbf{s^*} \in \secr^{l-1}}\secr_\mathbf{s^*}.
			$$
			Let $\mathbf{s^*} \in \secr^{l-1}$. In our analysis below we use the fact that any two $l$-tuples $\mathbf{s}=(s_{k+1},s_{k+2},\dots,s_{\chunknum})$ and $\mathbf{s'}=(s'_{k+1},s'_{k+2},\dots,s'_{\chunknum})$ in $\secr^l$ that appear in $\secr_\mathbf{s^*}$ have the property that all their coordinates differ. Namely that $s_{k+1}\ne s'_{k+1}, \dots, s_{\chunknum} \ne s'_{\chunknum}$.
			
			Now consider the set of $q^{nS}$ codeword suffixes $\rmegacodw$ corresponding to $l$-tuples $\mathbf{s}=(s_{k+1},s_{k+2},\dots,s_{\chunknum})$ from a certain set $\secr_\mathbf{s^*}$ in the partition specified above. 
			Each such codeword suffix consists of $l$ chunks.
			By our construction, the set of $q^{nS}$ 
			codeword suffixes corresponding to $\mathbf{s}=(s_{k+1},s_{k+2},\dots,s_{\chunknum}) \in \secr_\mathbf{s^*}$ are independent and uniformly distributed.
			This follows directly from our code construction and the property of  $\secr_\mathbf{s^*}$ discussed above.
			Thus, for $\mathbf{s}=(s_{k+1},s_{k+2},\dots,s_{\chunknum})$ and $\mathbf{s'}=(s'_{k+1},s'_{k+2},\dots,s'_{\chunknum})$ in $\secr_\mathbf{s^*}$, the event that a code suffix $\rmegacode$ is \textit{not} good with respect to message $m$, the list $\clist$, and the secrets $(s_{k+1},s_{k+2},\cdots,s_{\chunknum})$ is independent from the event that a code suffix $\rmegacode$ is \textit{not} good with respect to message $m$, the list $\clist$, and the secrets $(s^{\prime}_{k+1},s^{\prime}_{k+2},\cdots,s^{\prime}_{\chunknum})$.
			
			From Claim~\ref{claim:good-1}, a code suffix $\rmegacode$ is \textit{not} good with respect to message $m$, the list $\clist$, and a sequence of secrets $(s_{k+1},s_{k+2},\cdots,s_{\chunknum})$ with probability less than $q^{-(n-t)\delta}$. Thus, the probability that a code suffix $\rmegacode$ is \textit{not} good with respect to message $m$, the list $\clist$, and a certain $\sigma$ portion of sequences of $l$ secrets in the set $\secr_\mathbf{s^*}$ is less than
			\begin{align*}
			\left( q^{-(n-t)\delta}\right) ^{\sigma q^{nS}} = q^{-(n-t)\delta\sigma q^{nS}}.
			\end{align*}
			
			The number of all possible $\sigma$-portions of the set $\secr_\mathbf{s^*}$ is 
			\begin{align}
			\binom{q^{nS}}{\sigma q^{nS}} & < 2^{q^{nS}H_2(\sigma)}
			\nonumber \\
			& < 2^{q^{nS} \cdot (-2\sigma\log{\sigma})}.
			\label{eq:cl-gd-2-larger-3}
			\end{align}
			where \eqref{eq:cl-gd-2-larger-3} follows by $H_2(\sigma)<-2\sigma\log{\sigma}$ for $\sigma<1/2$.
			
			We say that a code suffix $\rmegacode$ is \textit{$\sigma$-good} with respect to message $m$, the list $\clist$ of codeword suffixes, and a secret set $\secr_\mathbf{s^*}$, if the code suffix $\rmegacode$ is good with respect to the message $m$, the list $\clist$, and a $(1-\sigma)$ portion of sequences of secrets in the set $\secr_\mathbf{s^*}$.
			So the probability over code design that a code suffix $\rmegacode$ is \textit{not} $\sigma$-good with respect to message $m$, list $\clist$, and secrets $\secr_\mathbf{s^*}$ is
			\begin{align}
			\mathbb{P}\left[ \rmegacode\text{ is not }\sigma\text{-good w.r.t.} \ m,\clist,\secr_\mathbf{s^*}\right] 
			& \leq q^{-(n-t)\delta \cdot \sigma q^{nS}} \cdot 2^{q^{nS} \cdot (-2\sigma\log{\sigma})}
			\nonumber
			\\
			& = q^{\sigma q^{nS} \left(-(n-t)\delta-2\log{\sigma}\log_q{2}\right)}
			\nonumber
			\\
			& \leq q^{\sigma q^{nS}\left( -\chunk \delta-2\log{\sigma}\log_q{2}\right) }
			\nonumber
			\\
			& = q^{q^{\left(n\zx{\echunk^3/q^2}-n\echunk^4\right)}\left( -n\zx{\echunk^3/q^2} + 2n\echunk^4 \right)}
			\label{eq:cl-gd-2-larger-4} 
			\\
			& < q^{-n^3}
			\label{eq:cl-gd-2-larger-5}
			\end{align}
			where \eqref{eq:cl-gd-2-larger-4} follows by substituting $\delta=\zx{\echunk^2/q^2}$, $S=\zx{\echunk^3/q^2}$, and $\sigma=q^{-n\echunk^4}$, and \eqref{eq:cl-gd-2-larger-5} follows for sufficiently large $n$.
			
			Now union bounding over all sets $\secr_\mathbf{s^*}$ in the partition of $\secr^l$, we get for sufficiently large $n$ that 
			\begin{align*}
			\mathbb{P}\left[ \exists \mathbf{s^*}:\ \rmegacode\text{ is not }\sigma\text{-good w.r.t.} \ m,\clist,\secr_\mathbf{s^*}\right]  
			\leq q^{-n^3}\cdot q^{nS(l-1)} < q^{-n^2}.
			\end{align*}
			
			Finally, we notice that being \textit{$\sigma$-good} with respect to a message $m$, a list $\clist$ of codeword suffixes, and any secret set $\secr_\mathbf{s^*}$ in the partition of $\secr^l$ implies being $\sigma$-good with respect to message $m$ and list $\clist$.
			Hence, the probability over code design that a code suffix $\rmegacode$ is $\sigma$-good with respect to message $m$ and list $\clist$ is
			\begin{align*}
			\mathbb{P}\left[ \rmegacode\text{ is }\sigma\text{-good w.r.t.}\ m, \clist \right] > 1-q^{-n^2}.
			\end{align*}
		\end{proof}
		
		\begin{remark}
			\label{re:coro-good-2}
			The goodness of a code suffix is what guarantees that the consistency check in the decoding process succeeds. Specifically, if a code is good with respect to a certain list and a certain message $m$; and in addition the codeword suffix received has few errors; then if message $m$ is in the list it will be (w.h.p.) the unique element that passes the consistency checking phase of Bob, and if it is not in the list the consistency checking phase of Bob will not return any message (w.h.p.).
		\end{remark}
		
		\begin{claim}
			\label{claim:good-3}
			Let $\sigma=\egood$. With probability greater than $1-q^{-n}$ over code design, for every message $m$, every list $\clist$, and every chunk end $t\in\setchend$, a code suffix is \textit{$\sigma$-good} with respect to message $m$ and list $\clist$.
		\end{claim}
		
		\begin{proof}
			The number of possible lists that can be obtained at a certain chunk end position $t$ depends on a set of messages of size $c/\erate$ for some constant $c$ and is thus at most of size 
			\begin{align}
			\label{eq:list-upper-bound}
			\binom{q^{nR}}{c/\erate} \leq q^{cnR/\erate}
			\end{align}
			
			From Claim~\ref{claim:good-2} we know that for $\sigma=\egood$, the probability that a code suffix $\rmegacode$ is $\sigma$-good with respect to all message $m$, any list $\clist$, and every chunk end position $\mt$ is at least 
			\begin{align*}
			1-q^{nR} \cdot q^{cnR/\erate} \cdot \chunknum \cdot q^{-n^2} & > 1-\zt{q^{-n^2+3cn/\erate}}
			\\
			& > 1-q^{-n}
			\end{align*} 
			for sufficiently large $n$.
		\end{proof}
		
		\subsection{Summary}
		
		\begin{claim}
			\label{claim:prob-good-exist}
			With probability at least $1-\frac{1}{n}-q^{-n}$ over code design, there exists a good code $\mcode$ such that the following properties are satisfied
			\begin{itemize}
				\item For any adversarial error \zt{and erasure} patterns, there exists a position $\tstar=\kstar\chunk$ such that the code prefix with respect to position $\tstar$, $\megacode[\kstar]$, is list decodable \zt{for $(\tstar-\lat[\tstar])\phat_{\tstar}$ errors} with list size $L=\blistord$ and that the transmitted message $m$ is in $\mlist$.
				Let $\clist$ be the list of codeword suffixes corresponding to $\mlist \setminus\{m\}$.
				
				\item \mikel{For any adversarial error and erasure patterns} and any position $t$ for which $t_0\leq t\leq\tstar$, the received word suffix with respect to position $t$ has \mikel{a total amount of erasures plus twice the amount of errors bounded by above by $(n-t)\left(\frac{q-1}{q}-\mikel{\frac{2\erate^2}{9q^{2}}}\right)$,  a total amount of errors bounded by $(n-t-n\pstar+\lat[t])\left(\frac{q-1}{2q}-\frac{\erate^2}{9q^{2}}\right)-\frac{n\pstar}{2q}$}, 
				and moreover the code suffix, $\rmegacode$, is $\sigma$-good with respect to the transmitted message $m$ and the list $\clist$ where $\sigma=\egood$.
			\end{itemize}
		\end{claim}
		
		\begin{proof}
			We consider all possible error \zt{and erasure} patterns of the adversary by analyzing all of Calvin's possible trajectories. \zt{More precisely, given any erasure pattern, we analyze Calvin's possible behaviors $\ppp_{\mt}$ on the $(\mt-\lat)$ unerased symbol positions.} As mentioned above, all possible trajectories of Calvin can be classified into two types, the High Type Trajectory and the Low Type Trajectory.
			
			For any Low Type Trajectory, we have $\ppp_{t_0}<\phat_{t_0}$. Let $t_0=k_0\chunk$ for some integer $k_0$. Notice that by our choice of $\phatt$, the list-decoding condition \eqref{eq:list-con-c} is always satisfied. Therefore, by Corollary~\ref{coro:b-list-decodable}, with list decoding radius $(\mt_{0}-\lat[\mt_{0}])\phat_{t_{0}}$, the code prefix, $\megacode[k_{0}]$, is list decodable \zt{for errors} with list size $\blistord$ with probability $1-\frac{1}{n}$ over code design. In addition, since $(\mt_{0}-\lat[\mt_{0}])p_{t_0}<(\mt_{0}-\lat[\mt_{0}])\phat_{t_0}$, we have $m\in\mlist$. So far the first property stated in the claim is satisfied for any Low Type Trajectory.
			
			By Claim~\ref{claim-zero-starting}, \zx{$\ppp_{\mt_{0}}$ satisfies the energy bounding condition \eqref{eq:energy-con-c} and by Definition~\ref{def:p-tilde-t}, we have $\ppp_{\mt_{0}}\geq\ptil_{\mt_{0}}$. Then by Claim~\ref{claim:p-tilde-t}} the received word suffix with respect to position $\mt_{0}$ has no more than \mikel{a fraction of} $\zx{\frac{q-1}{2q}-\frac{\erate^2}{9q^{2}}-\frac{n\pstar}{2q(n-\mt_{0}-n\pstar+\lat[\mt_{0}])}}$ of \zt{its unerased symbols in error}. 
			\zx{Moreover, since there are \mikel{at most} $n\pstar-\lat[\mt_{0}]$ erasures in the received word suffix, we have \mikel{that the total amount of erasures and twice the amount of errors in the suffix is} $n\pstar-\lat[\mt_{0}]+(n-\mt_{0}-n\pstar+\lat[\mt_{0}])\left(\frac{q-1}{q}-\frac{2\erate^2}{9q^{2}}-\frac{n\pstar}{q(n-\mt_{0}-n\pstar+\lat[\mt_{0}])}\right)
			 < (n-\mt_{0})\left(\frac{q-1}{q}-\mikel{\frac{2\erate^2}{9q^{2}}}\right)$.}
			By Claim~\ref{claim:good-3}, the code suffix $\rmegacode[k_{0}]$ is $\sigma$-good with respect to message $m$ and list $\clist$ with probability $1-q^{-n}$ over code design. Hence, for any Low Type Trajectory, our code design possesses the two properties stated in the claim. Moreover, in this case we have $\tstar=t_0$.
			
			For any High Type Trajectory, we have $p_{t_0} \geq \phat_{t_0}$. By Claim~\ref{claim-above}, given any trajectory $\pref{\ppp}$ of High Type, the trajectory $\pref{\ppp}$ always intersects with $\phatt$ no later than the position $\mt = \lat+n-\frac{q}{q-1}n\pstar-\chunk$. 
			Let $\tstar$ be the chunk end immediately after the intersection point, at which $\ppp_{\tstar}\leq\phat_{\tstar}$ (which implies $\ppp_{\tstar-\chunk}\geq\phat_{\tstar-\chunk}>\ptil_{\tstar-\chunk}$). 
			Let $t=k\chunk\leq\tstar$. Then at any position $t$, by Corollary~\ref{coro:b-list-decodable}, with list decoding radius $\zt{(\mt-\lat)}\phat_{t}$, the code prefix $\megacode[k]$ is list decodable \zt{for errors} with list size $\blistord$ with probability $1-\frac{1}{n}$ over code design. Also, for $\tstar$, since $(\tstar-\lat[\tstar]) p_{\tstar}<(\tstar-\lat[\tstar])\phat_{\tstar}$, the transmitted message $m$ is in the list $\mlist$. 
			
			Since \mikel{$\ppp_{\tstar-\chunk}>\phat_{\tstar-\chunk}$}, then by Claim~\ref{claim-ingap} we have $\ppp_{\tstar}>\ptil_{\tstar}$, and further, by Claim~\ref{claim:p-tilde-t}, for any trajectory $\ppp_{t}$ of High Type, if $t \leq \tstar$ then the received word suffix with respect to position $t$ has no more than \zt{a fraction of} $\zx{\frac{q-1}{2q}-\frac{\erate^2}{9q^{2}}-\frac{n\pstar}{2q(n-\mt_{0}-n\pstar+\lat[\mt_{0}])}}$ of its \zt{unerased symbols in error}. 
			\zx{\mikel{As above}, we have $n\pstar-\lat+(n-\mt-n\pstar+\lat)\left(\frac{q-1}{q}-\frac{2\erate^2}{9q^{2}}-\frac{n\pstar}{q(n-\mt-n\pstar+\lat)}\right) < (n-\mt)\left(\frac{q-1}{q}-\mikel{\frac{2\erate^2}{9q^{2}}}\right)$}. By Claim~\ref{claim:good-3} the code suffix with respect to position $t$, $\rmegacode[k]$, is $\sigma$-good with respect to message $m$ and list $\clist$ with probability $1-q^{-n}$ over code design. Thus far, for any High Type Trajectory, both the properties in the claim are also satisfied by our code design.
			
			In conclusion, the probability that the code $\mcode$ possesses the two properties is at least $1-\frac{1}{n}-q^{-n}$.
		\end{proof}
		
		\begin{remark}
			\label{rem:decode}
			Note that, using the code from Claim~\ref{claim:prob-good-exist}, the position $\tstar$ can found by Bob through an iterative decoding process starting from the position $t_{0}$, and therefore, the decoding process of Bob can stop at some $\tstar$ correctly. More precisely, Claim~\ref{claim:prob-good-exist} ensures that every time Bob obtains a list of codewords, then no matter if the transmitted message $m$ is in the list $\mlist$ or not, the code suffix with respect to position $t \leq \tstar$ is $\sigma$-good with respect to message $m$ and the list $\clist$ of codeword suffixes. In other words, if $t$ is strictly smaller than $\tstar$ then the consistency decoding of Bob will not return any message, and when $t=\tstar$ the consistency decoding will return the correct message (all with high probability over the randomness of Alice). 
			Thus, Bob can correctly determine whether to continue the decoding process or not.
		\end{remark}
		
		\begin{claim}
			\label{claim:prob-decode}
			Let $\alfq=\alfexpq$ where $\pbar\in\left[0,\ppp\right]$. Let $$C=\capopq$$ and $R=C-\erate$. 
			For any message $m\in\msg$ and its corresponding encoding $\mathbf{x}\in\mathcal{X}^{n}$ using the code established in Claim~\ref{claim:prob-good-exist} and the encoder of Section~\ref{sec:lower-bound}, the decoding procedures described in Section~\ref{sec:lower-bound} allows Bob to correctly decode the message $m$ with probability at least $1-n\egood$ over the random secrets $s\in\secr$ available to Alice. 
		\end{claim}
		
		\begin{proof}
			A decoding error occurs if the consistency decoder fails to return a single message or if the decoder returns a message that is not equal to the transmitted message. 
			For all $t$ strictly less than $\tstar$ of Claim~\ref{claim:prob-good-exist}, we have by property (2) of Claim~\ref{claim:prob-good-exist}, Remark~\ref{rem:decode}, and by the definition of Step (3) of our decoding procedure that the consistency check in the decoding process will not return any message (with probability $1-\sigma$ over the randomness of the encoding). 
			\zx{More precisely, by Definition~\ref{def:t-star} and the definition of our iterative decoding process, for any $\mt$ strictly less that $\tstar$, we have $\ppp_{\mt}>\phatt$. Then since our list-decoding radius is $\mt\phatt<\mt\ppp_{\mt}$, the list we obtain from the list-decoding phase will not include the transmitted message and the consistency decoder will not return any message with high probability.}
			In addition, for $t=\tstar$, with the same probability, the consistency check of the decoding process will return the correct message.
			\zx{Specifically, for $\mt=\tstar\neq\mt_{0}$, by Claim~\ref{claim-ingap} we have $\ppp_{\tstar}\geq\ptil_{\tstar}$. 
			For $\mt=\tstar=\mt_{0}$, by Claim~\ref{claim-zero-starting}, we have the energy bounding condition satisfied by $\ppp_{\mt_{0}}$, and by Definition~\ref{def:p-tilde-t}, we have $\ppp_{\tstar}\geq\ptil_{\tstar}$. 
			As the energy bounding condition is satisfied at $\tstar$ and $\ppp_{\tstar}\geq\ptil_{\tstar}$, we have by Claim~\ref{claim:p-tilde-t}, the amount of errors in the codeword suffix is bounded, and therefore, by the definition of our consistency decoder and Claim~\ref{claim:good-3}, the consistency decoder will return the correct message with high probability.}
			In both cases, the success probability is obtained by the probability that the sequence of $l$ secrets used in the codeword suffix is not chosen from the particular $\sigma$ portion of $\secr^{l}$ that may cause a decoding failure. 
			
			From Claim~\ref{claim:good-3}, we have $\sigma=\egood$.
			Therefore, the probability of successful decoding is at least
			\begin{align*}
			1-n\sigma=1-n\egood.
			\end{align*}
		\end{proof}
		
		\begin{theorem}
			\label{thm:capacity}
			The capacity $C$ of $q$-ary causal adversarial channels with symbol errors and erasures is
			\begin{align}
			\capopq
			\label{eq:cap-q-ary}
			\end{align} where $\alfq=\alfexpq$.
		\end{theorem}
		
		\begin{proof}
			Let $\xi>0$ and $\beta>0$.
			The converse is proven in Section~\ref{sec:upper-bound}. Namely, for any code $\mcode$ with stochastic encoding of rate $R=C+\beta$, the average error probability is lower bounded by $\beta^{O\left(\frac{1}{\beta}\right)}$.
			The achievability proof follows from Claim~\ref{claim:prob-decode} in Section~\ref{sec:lower-bound}. Specifically, for sufficiently large $n$ it holds by Claim~\ref{claim:prob-decode} that the decoding error is bounded above by $\xi$.
			In addition, for sufficiently small $\epsilon$, by the continuity of the $q$-ary entropy function, the code rate $R=C-\erate$ of Claim~\ref{claim:prob-decode} is at least $C - \beta$.
			Therefore, for sufficiently large $n$, $q^{nR}=q^{n\left(C-\beta\right)}$ distinct messages can be reliably transmitted over our channel with error probability at most $\xi$. Hence, the channel capacity of $q$-ary causal adversarial channels with symbol errors and erasures is $C$.
		\end{proof}
		
		\section{Discussion of Special Cases}
		In this section, we discuss several special cases of $q$-ary causal adversarial channels.
		\subsection{Symbol Error Channel}
		For $q$-ary causal adversarial channels with symbol errors only, the above analysis can get modified by setting $\pstar=0$ and $\lat=0$ to obtain the corresponding \zx{capacity:
		\begin{align*}
		\capopq
		\end{align*} where $\alfq=1-\frac{2q}{q-1}\left(p-\bar{p}\right)$.}
		\subsection{Symbol Erasure Channel}
		For $q$-ary causal adversarial channels with erasures only, there is no need for a decoding reference trajectory $\phatt$ since erasures are visible. The corresponding list-decoding condition becomes
		\begin{align}
		\label{eq:list-con-e}
		\mt-\lat-\frac{n\erate}{4} \geq nR.
		\end{align}
		It can be shown that there exists $\mt\in\setchend$ and $(\mt-\lat)\in
		\left[ n\left(1-\frac{q}{q-1}\pstar-\frac{\erate^2}{4}\right),n\left(1-\frac{q}{q-1}\pstar-\frac{\erate^2}{9(q-1)}\right)\right]$ such that the following energy-bounding condition is satisfied.
		\begin{align}
		\label{eq:energy-con-e}
		n\pstar-\lat+\frac{(n-t)\erate^2}{9q^2}\leq\frac{q-1}{q}(n-\leftcw)
		\end{align}
		\zt{With these modified conditions, the decoder Bob can pin-point the value of $\tstar$ for which the modified conditions are satisfied, and therefore, Bob is also able to determine his list decoding radius to be $\lat[\tstar]$.} The corresponding capacity is $$1-\frac{q}{q-1}\pstar.$$
		\subsection{Large Alphabet}
		For sufficiently large $q$, we have $\alfq\approx 1-2(p-\pbar)-\pstar$ and $H_{q}\left(\frac{\pbar}{\alfq}\right)\approx\frac{\pbar}{\alfq}$. Then we obtain
		\begin{align*}
		C& = \capopq\\
		& \approx \min_{\pbar\in\left[0,p\right]}\left[\alfq\left(1-\frac{\pbar}{\alfq}\right)\right]\\
		& = \min_{\pbar\in\left[0,p\right]}\left[\alfq-\pbar\right]\\
		& \approx \min_{\pbar\in\left[0,p\right]}\left[1-2p-\pstar+\pbar\right]\\
		& = 1-2p-\pstar
		\end{align*}
		\zx{Hence, for sufficiently large alphabets}, if the adversary has no erasure \zx{budget}, i.e., $\pstar=0$, the capacity is $1-2\ppp$, which matches the one given in \cite{dey2013codes}. On the other hand, if the adversary only has erasure \zx{budget}, i.e., $\ppp=0$, the capacity is $1-\pstar$.
		
		\zt{We also depict some of the special cases discussed above in Figure~\ref{fig:cap-era-err}, and a comparison of the binary online setting with other bounds in Figure~\ref{fig:bounds}.}
		
		\begin{figure}[htb]
			\centering
			\begin{subfigure}{.5\textwidth}
				\centering
				\includegraphics[width=1\linewidth]{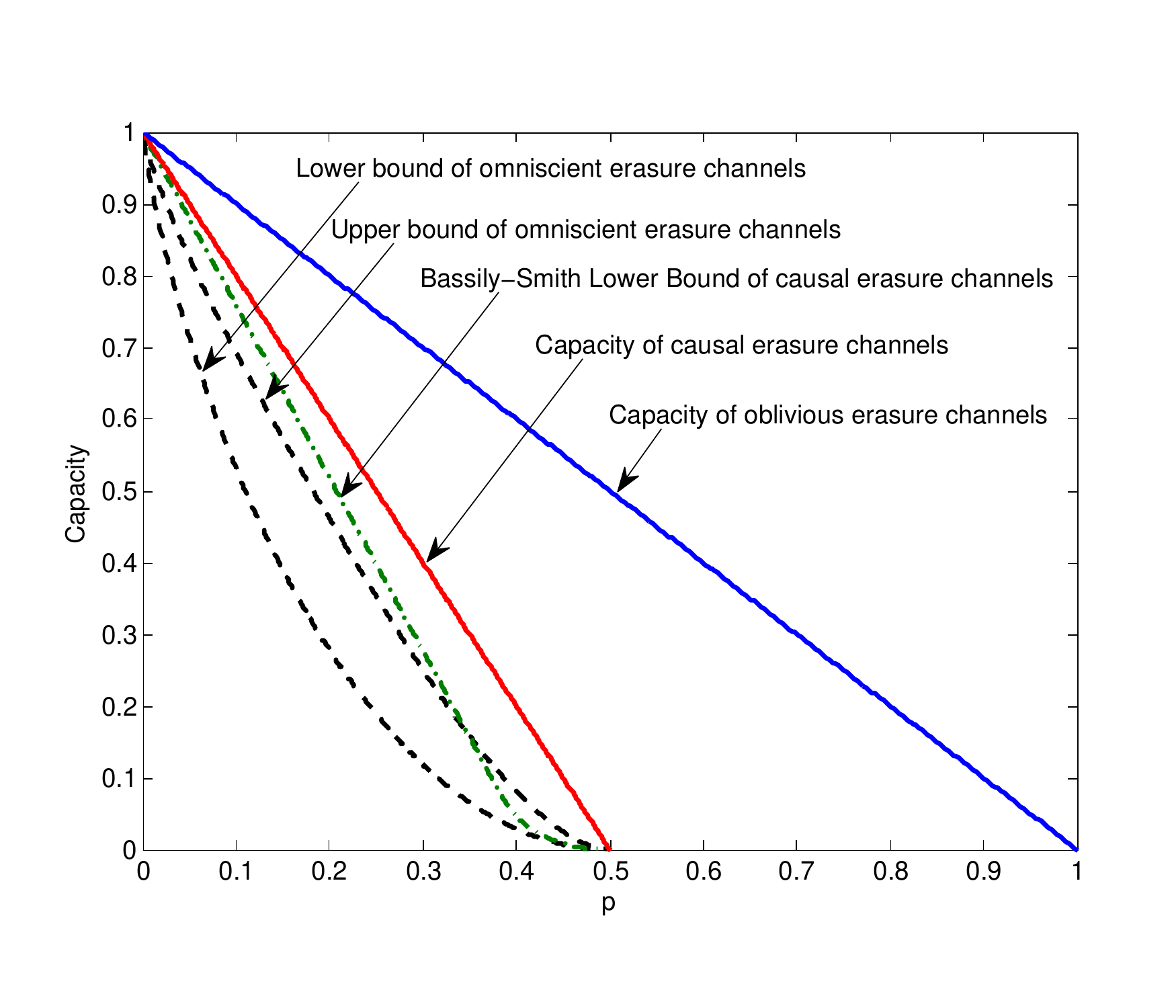}
				\caption{Binary adversarial erasure channels: The bound of $1-p$ (in blue) corresponds to the capacity of binary oblivious erasure channel. The MRRW bound and the GV bound (both in dotted black) are the best known upper and lower bounds for binary omniscient erasure channels. The lower bound for binary causal erasure channels by Bassily and Smith~\cite{bassily2014causal} is plotted in green.} 
				\label{fig:bound-bec}
			\end{subfigure}%
			\begin{subfigure}{.5\textwidth}
				\centering
				\includegraphics[width=1\linewidth]{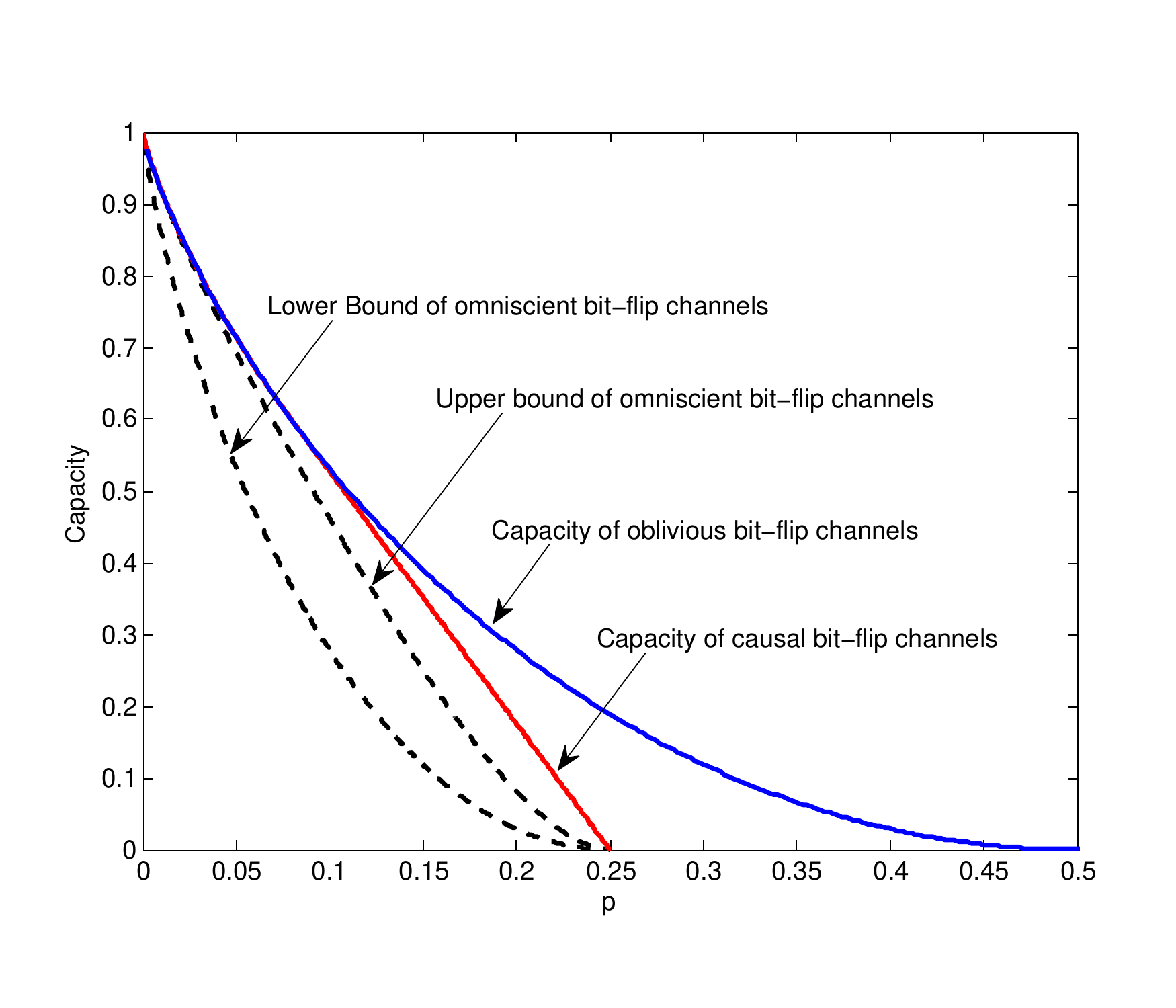}
				\caption{Binary adversarial bit-flip channels: The bound of $1-H(p)$ (in blue) corresponds to binary oblivious bit-flip channel. The MRRW bound and the GV bound are the best upper and lower bounds (both in dotted black) for binary omniscient bit-flip channels. For binary causal bit-flip channels, the previous lower bound by Haviv and Langberg~\cite{haviv2011beating} is a slight improvement over the GV bound.} 
				\label{fig:bounds-bsc}
			\end{subfigure}
			\caption[Bounds on the capacity of binary online adversarial channels]{Bounds on the capacity of binary online adversarial channels}
			\label{fig:bounds}
		\end{figure}
		
		\begin{figure}[htb]
			\centering
			\begin{subfigure}{.5\textwidth}
				\centering
				\includegraphics[width=1\linewidth]{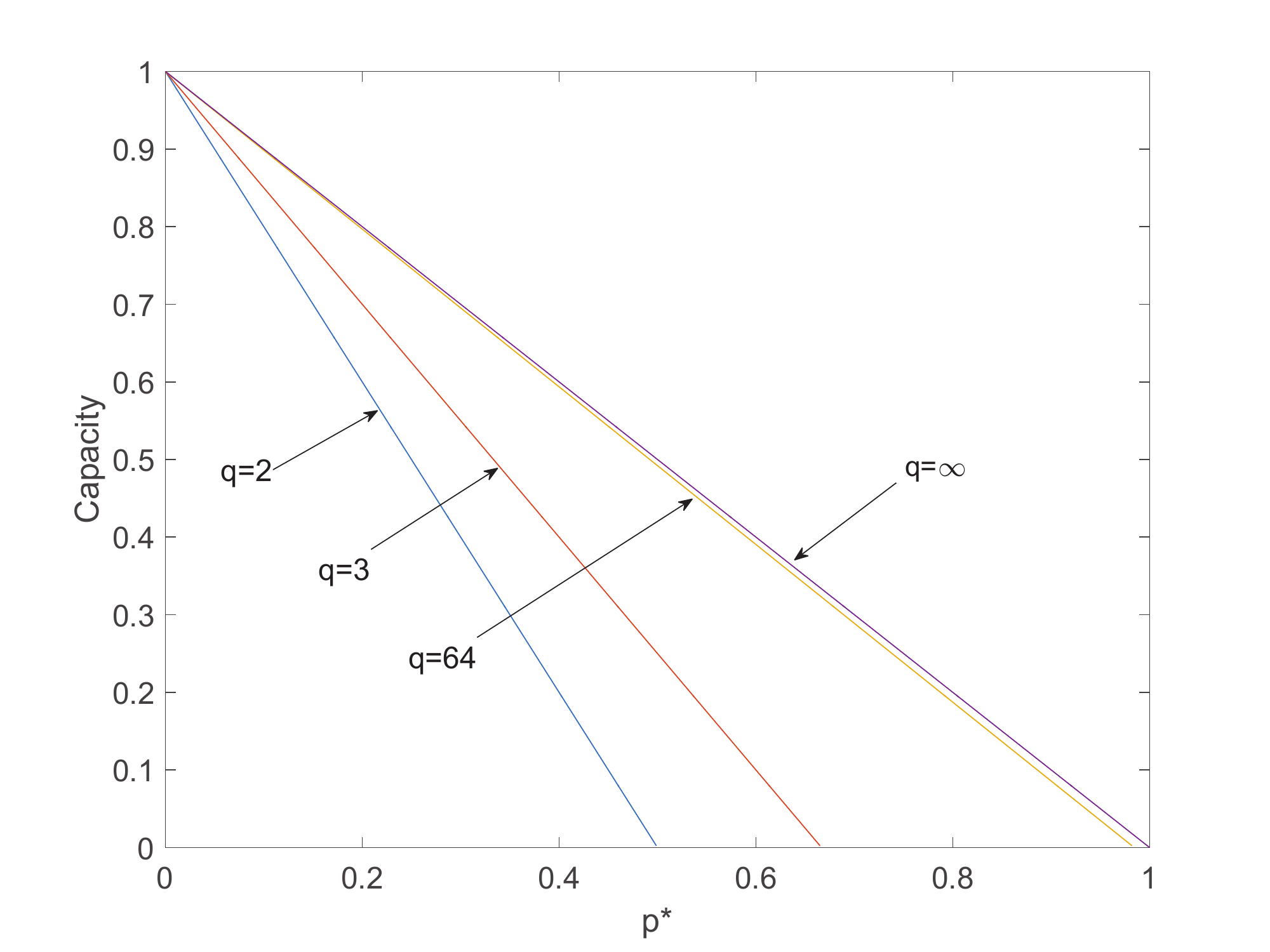}
				\caption{Online $q$-ary erasure channels} 
				\label{fig:cap-erasures}
			\end{subfigure}%
			\begin{subfigure}{.5\textwidth}
				\centering
				\includegraphics[width=1\linewidth]{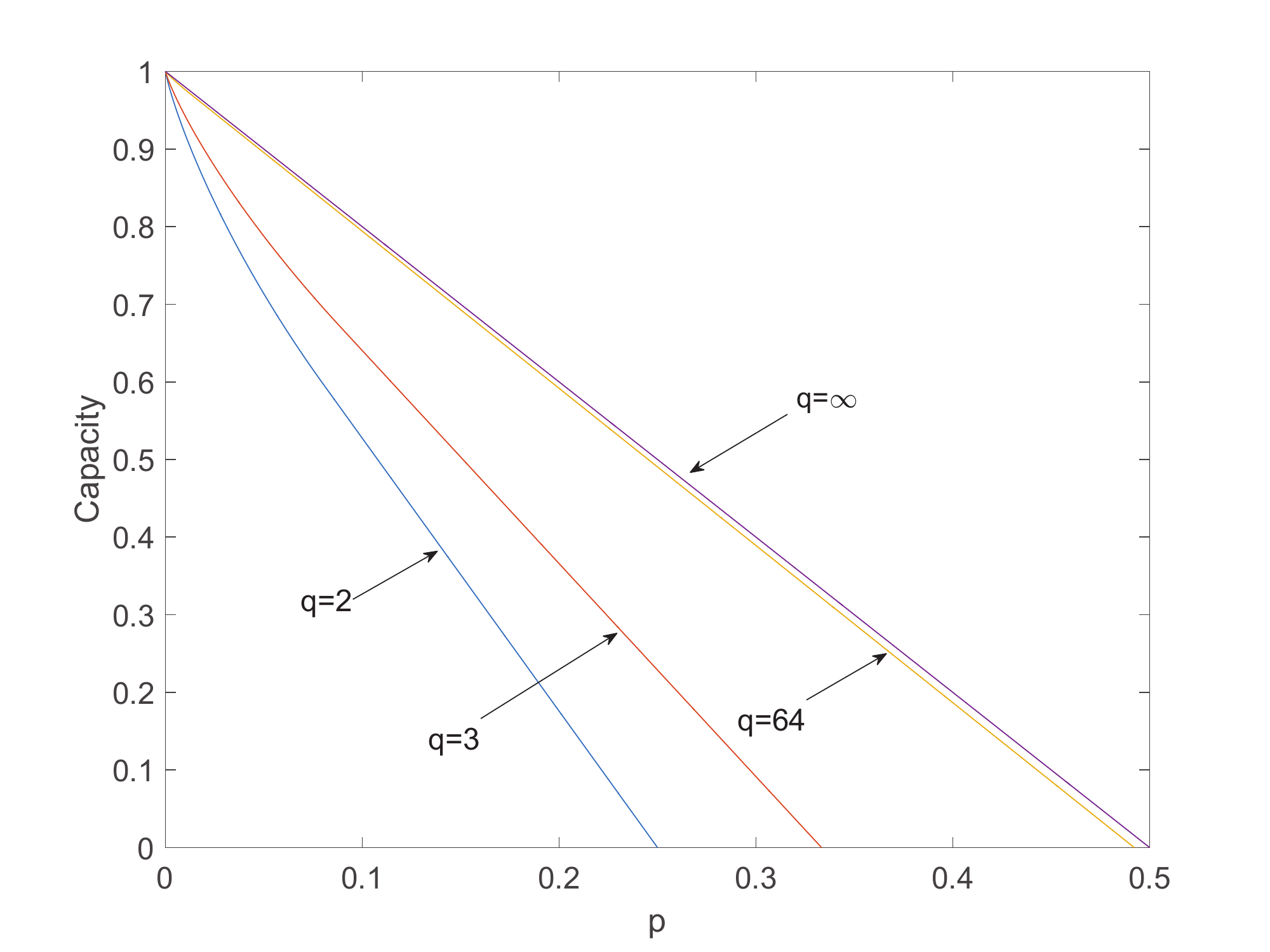}
				\caption{Online $q$-ary error channels}
				\label{fig:cap-errors}
			\end{subfigure}
			\begin{subfigure}{.5\textwidth}
				\centering
				\includegraphics[width=1\linewidth]{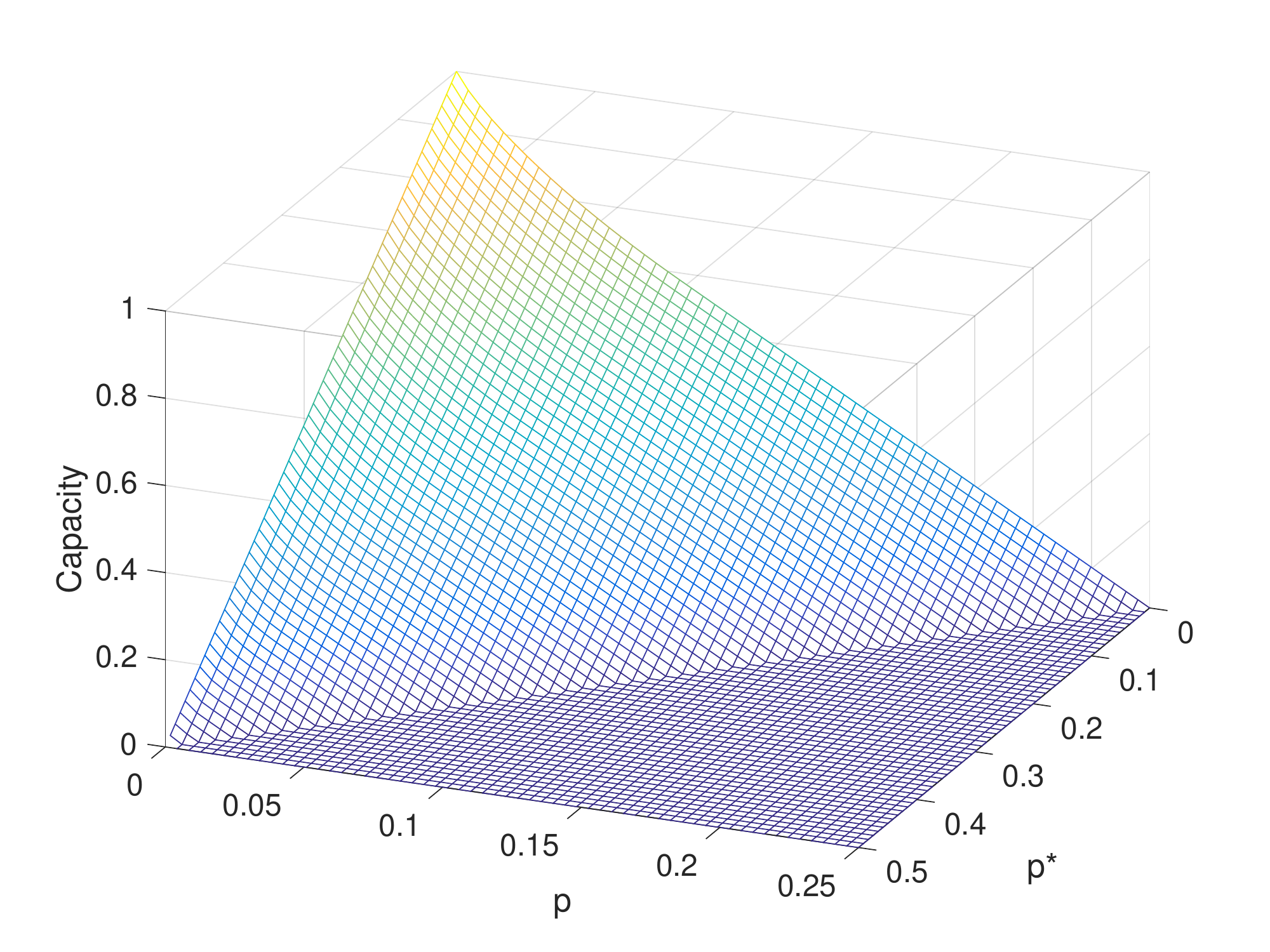}
				\caption{Online binary error-erasure channels} 
				\label{fig:cap-error-erasure-2}
			\end{subfigure}%
			\begin{subfigure}{.5\textwidth}
				\centering
				\includegraphics[width=1\linewidth]{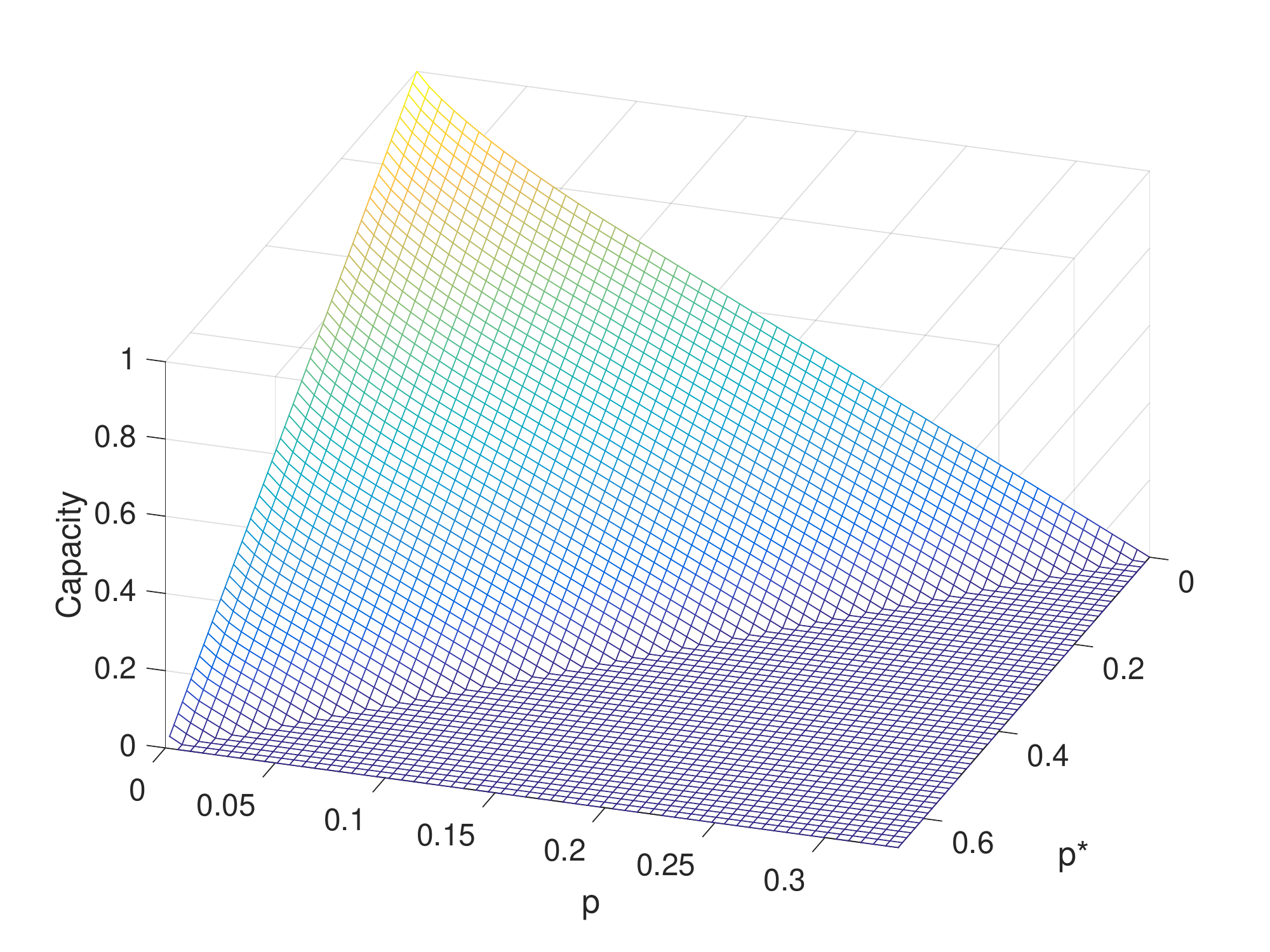}
				\caption{Online ternary error-erasure channels}
				\label{fig:cap-errors-erasure-3}
			\end{subfigure}
			\caption[Capacity of $q$-ary channels]{Capacity for a number of online $q$-ary channels}
			\label{fig:cap-era-err}
		\end{figure}
		
		\newpage
		\begin{table}[htbp]
			\renewcommand{\arraystretch}{1.4}
			\centering
			\caption{Table of Parameters}
			\begin{tabular}{|c|l|c|}
				\hline
				symbol & description & equality/range \\ \hline
				
				$C$ & capacity & $\eqref{eq:cap-q-ary}$ \\ \hline
				
				$n$ & block length & \\ \hline
				$\ppp$ & fraction of a codeword that can be changed & $\left(\intverr\right)$ \\ \hline
				$\pstar$ & fraction of a codeword that can be erased & $\left(\intvera\right)$ \\ \hline
				$\echunk$ & ``quantization'' parameter & $\frac{\erate^2}{9q^2}$ \\ \hline
				
				$R$ & code rate & $C-\erate$ \\ \hline
				$S$	& private secret rate & $\zx{\echunk^3/q^2}$ \\ \hline
				
				$\msg$ & message set & $\msg=\left[q^{nR}\right]$\\ \hline
				$\secr$ & secret set & $\secr=\left[q^{nS}\right]$ \\ \hline
				$\mathcal{X}$ & input alphabet & $\qaryset$ \\ \hline
				$\mathcal{Y}$ & output alphabet & $\qaryset\cup\{\Lambda\}$ \\ \hline
				
				$\mathcal{T}$ & set of chunk ends & $\choiceschunk$ \\ \hline

				$\mbfu$ & random variable of input message & \\ \hline
				$\mbfx$ & random variable of input codeword & \\ \hline
				$\mbfy$ & random variable of output word & \\ \hline
				
				$m$ & message & $m\in\msg$\\ \hline
				$\mathbf{x}$ & codeword & $\mathbf{x}\in\mathcal{X}^{n}$ \\ \hline
				$s$ & secret & $s\in\secr$\\ \hline
				$\mathbf{s}$ & secret & $\mathbf{s}\in\secr^{n}$\\ \hline
				
				$\mt$ & length of prefix & $t\in\setchend$ \\ \hline
				$\lat$ & number of erasures up to position $\mt$ & \\ \hline
				$k$ & number of chunks in the prefix w.r.t. position $t$ & $k=\frac{t}{\chunk}$\\ \hline
				$l$ & number of chunks in the suffix w.r.t. position $t$ & $l=\chunknum-k$\\ \hline
				
				$\pref{\ppp}$ & adversary's trajectory & \\ \hline
				$\pbar_{t}$ & guess of random noise & \eqref{eq:p-bar-t} \\ \hline
				$\phatt$ & decoding reference trajectory & \eqref{eq:p-hat-t} \\ \hline
				$\tilde{p}_{t}$ & energy bounding trajectory & \eqref{eq:p-tilde-t} \\ \hline

				$\mlist$ & a list of messages & \\ \hline
				$\clist$ & a list of codeword suffixes excluding suffixes corresponding to $m$ & \\ \hline
				$\mlists$ & list size of $\mlist$ & $\blistord$ \\ \hline
				$\clists$ & list size of $\clist$ & $q^{nSl} \cdot \blistord$ \\ \hline
				
			\end{tabular}
			\label{tab:q-params}
		\end{table}
\end{appendices}
	
\end{document}